\documentclass[11pt, a4paper, onecolumn, confidential, copyright, gdm]{google}

\uselogo{gdm} 

\usepackage{amsmath}
\usepackage{microtype}
\usepackage{graphicx}
\usepackage{booktabs} 

\usepackage{pifont}

\usepackage{caption}
\usepackage{subcaption}
\usepackage[cc]{titlepic}

\usepackage{amsmath}
\usepackage{amsthm}
\usepackage{amsbsy,amsfonts,amssymb}
\usepackage{bm}
\usepackage{multirow}
\usepackage{wrapfig}
\usepackage{thmtools}
\usepackage{cleveref}
\usepackage{xfrac}

\usepackage{tgpagella}
\usepackage[utf8]{inputenc} 
\usepackage[T1]{fontenc}    
\usepackage{authblk}
\usepackage{algorithm}
\usepackage{algorithmic}
\usepackage[authoryear, sort&compress, round]{natbib}

\usepackage[export]{adjustbox}
\usepackage{makecell}

\usepackage{nicefrac}       
\usepackage[svgnames]{xcolor}         

\usepackage{enumitem}
\usepackage{rotating}   
\usepackage{tablefootnote}
\usepackage[title]{appendix}

\usepackage{mathtools}
\usepackage{url}  
\usepackage{hyperref} 
\usepackage{cleveref} 

\usepackage{mdframed} 

\usepackage[utf8]{inputenc} 
\usepackage[T1]{fontenc}    
\usepackage{hyperref}       
\usepackage{url}            
\usepackage{booktabs}       
\usepackage{amsfonts}       
\usepackage{nicefrac}       
\usepackage{microtype}      
\usepackage{xcolor}         

\usepackage{tabularx}
\usepackage{booktabs}   
\usepackage{multirow}   
\usepackage{xcolor}     
\usepackage{colortbl}   

\usepackage{listings} 

\usepackage{graphicx} 

\definecolor{lightred}{rgb}{1,0.9,0.9} 
\definecolor{lightgreen}{rgb}{0.9,1,0.9} 

\usepackage{tikz}
\usetikzlibrary{trees}
\usepackage{forest}
\usepackage{mathtools}

\usepackage[mathscr]{eucal}

\newcommand{\Edocids}{\operatorname{E}_{\text{docIDs}}}
\newcommand{\Eprime}{\operatorname{E}^{\prime}}
\newcommand{\Vdocids}{\mathscr{V}_{\text{docIDs}}}

\usepackage{tcolorbox}
\tcbuselibrary{skins}
\tcbset{
    promptbox/.style={
        arc=5mm,
        boxrule=0.5pt,
        colframe=black,
        colback=white,
    }
}

\declaretheoremstyle[
  spaceabove=\topsep, spacebelow=\topsep,
  headfont=\normalfont\bfseries,
  notefont=\mdseries, notebraces={(}{)},
  bodyfont=\normalfont\itshape,
  postheadspace=\newline,
]{mythmstyle}

\declaretheorem[style=mythmstyle, name=Proposition, numberwithin=section]{prop}
\declaretheorem[style=mythmstyle, name=Corollary, numberwithin=section]{corollary}

\declaretheorem[style=definition, numberwithin=section]{definition}

\newcolumntype{C}[1]{>{\centering\arraybackslash}p{#1}}

\title{Autoregressive Ranking: Bridging the Gap Between Dual and Cross Encoders}
\date{}

\author[$\pi$]{Benjamin Rozonoyer\footnote{Work done during internship at Google Research.}}
\author[$\lambda$]{Chong You}
\author[$\lambda$]{Michael Boratko}
\author[$\lambda$]{Himanshu Jain}
\author[$\theta$]{Nilesh Gupta\protect\footnotemark[1]}
\author[$\lambda$]{Srinadh Bhojanapalli}
\author[$\lambda$]{Andrew McCallum}
\author[$\lambda$]{~Felix Yu}

\affil[$\pi$]{University of Massachusetts Amherst}
\affil[$\lambda$]{Google DeepMind}
\affil[$\theta$]{The University of Texas at Austin}

\begin{abstract}
     The success of Large Language Models (LLMs) has motivated a shift toward generative approaches to retrieval and ranking, aiming to supersede classical Dual Encoders (DEs) and Cross Encoders (CEs). A prominent paradigm is pointwise Autoregressive Ranking (ARR), where an LLM generates document identifiers (docIDs) token-by-token to enable ranking via beam search. ARR offers the promise of superior expressivity compared to DEs while avoiding the prohibitive computational cost of CEs. However, a formal theoretical foundation for this expressive power has been missing. Moreover, the standard next-token prediction loss is rank-agnostic and inappropriate for finetuning an LLM for ranking tasks.

    In this paper, we first prove that the expressive capacity of ARR is strictly superior to DEs. While a DE requires an embedding dimension that grows linearly with corpus size to achieve arbitrary rankings, ARR can solve it with a constant hidden dimension. We then propose \textsc{SToICaL} (\underline{S}imple \underline{To}ken-\underline{I}tem \underline{Ca}librated \underline{L}oss), a generalized rank-aware training loss for LLM finetuning. By using item-level reweighting and prefix-tree marginalization, we distribute probability mass over valid docID tokens based on their ground-truth relevance. Experiments on WordNet and ESCI datasets verify that our loss suppresses invalid docID generations and significantly improves ranking metrics beyond top-1 retrieval.
\end{abstract}

\begin{document}

\maketitle

\begin{figure}[h] %
    \centering %
    \includegraphics[width=\linewidth]{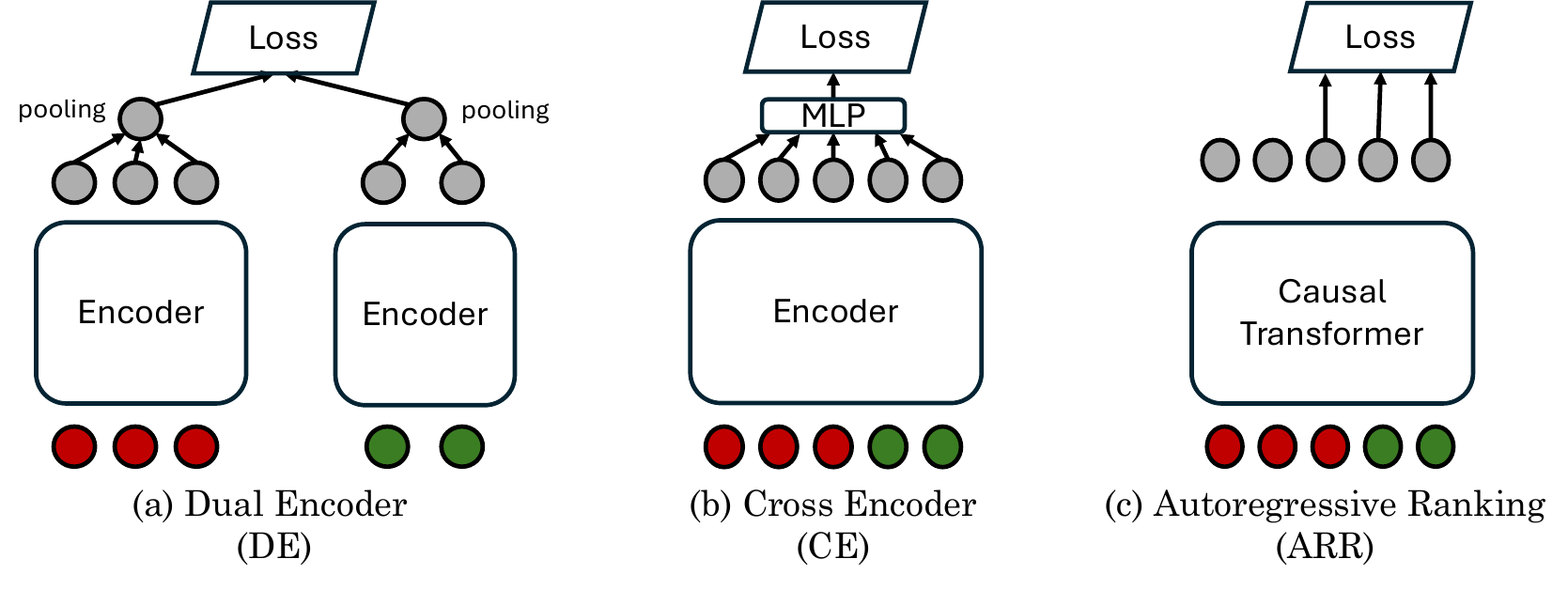}
    \caption{Illustration of different architectures for ranking. Red and green circles represent query and document tokens, respectively. 
    (a) DEs are efficient but have limited expressive power. 
    (b) CEs produce more accurate relevance scores but have high computational cost. 
    Hence, the standard practice is a two-stage pipeline where a DE selects a candidate set and a CE reranks those candidates.
    (c) ARR holds the promise of a unified model for ranking that replaces the DE/CE pipeline.}
    \label{fig:t3-xenc-pt}
\end{figure}

\section{Introduction}\label{sec:introduction}

Ranking documents by their relevance to a user query is a cornerstone of modern information retrieval. 
The dominant paradigm for this task is a multi-stage pipeline. 
The first stage, retrieval, aims to efficiently select a candidate set of documents from a massive corpus. 
This is typically accomplished using \textbf{Dual Encoder (DE)} models, which represent queries and documents as dense vectors and use fast approximate nearest neighbor (ANN) search to find the top matches \citep{karpukhin2020dense,guo2020accelerating}. While efficient, the expressive power of DEs is inherently limited, as the rich interaction between a query and a document is compressed into a single vector similarity.  Consequently, a second reranking stage is often required. In this stage, a more powerful but computationally expensive \textbf{Cross Encoder (CE)} model jointly processes the query and each candidate document to produce a more accurate relevance score \citep{nogueira2019passage}. CEs are infeasible for first-stage retrieval due to their prohibitive cost, which scales linearly with corpus size.

The success of LLMs has inspired a new, unified approach to retrieval and ranking in lieu of the two-stage approach \citep{metzler2021rethinking,li2023large,lee2024can}. 
In this paradigm, an LLM is trained to function as a generative retrieval model, autoregressively decoding the relevant documents or their identifiers (docIDs) token-by-token, conditioned on the query \citep{tay2022transformer}. For ranking, this model can generate a sorted list of documents by employing beam search, effectively ranking documents by their conditional generation probability. 
This approach, which we dub \textbf{Autoregressive Ranking (ARR)}, offers several compelling advantages over the traditional DE/CE pipeline. It collapses the two-stage process into a single model, eliminating the need for a separate ANN index. Furthermore, it is more efficient than CEs as it does not require scoring every document individually, while preserving CEs' cross-attention over already generated docID tokens. Crucially, it has also been argued that \textit{by generating documents token-by-token}, ARR models can capture deeper and more nuanced query-document interactions than DEs \citep{yuan2024generative}; while intuitive, the evidence for this has been primarily empirical. In this paper, we address two critical gaps in the understanding and optimization of ARR models for ranking:

\begin{itemize}[align=left,leftmargin=*,itemsep=0pt,topsep=0pt]
    \item \textbf{We provide a theoretical foundation for the superior expressive capacity of ARR over DEs.} A rigorous analysis of the embedding geometry required for ranking shows that for a DE to achieve any ordering of $k$ documents, its embedding dimension must grow linearly with $k$. In contrast, we prove that an ARR model with constant hidden dimension is theoretically sufficient to rank an arbitrary number of documents. This offers a formal explanation for the advantages of ARR.
    \item \textbf{We introduce \textsc{SToICaL}, a rank-aware generalization of next-token prediction loss for enhanced ranking capabilities.} Most existing generative retrieval methods are trained with standard next-token prediction, an objective which fails to explicitly model the relative ordering of documents, which is crucial for ranking. We introduce a new loss function designed to make the LLM aware of the entire ranked list during training, while keeping the generation and signal length on the order of a single docID rather than a full list, enabling ranking via beam search. 
    We test our method on WordNet and a shopping queries dataset, demonstrating its superior ranking performance. 
\end{itemize}

\section{Related Work}\label{sec:related-work}

\paragraph{Classical IR Losses} Training objectives for retrieval and ranking can be broadly grouped into the categories of \textit{pointwise}, \textit{pairwise}, and \textit{listwise} \citep{liu2009learning}. A \textit{pointwise} loss as in \citet{chen2009ranking} of the form $\mathcal{L}(q, d)$ encourages the model to assign high scores to more relevant items and vice versa; each training example consists of a single document to penalize independently of the rest. \textit{Pairwise} loss functions \citep{cao2007learning} of the form $\mathcal{L}(q, d^{+}, d^{-})$ strive to maximize the margin between relevant and irrelevant documents. \textit{Listwise} losses such as \cite{xia2008listwise} of the form $\mathcal{L}(q, d_1, \ldots, d_{n_q})$ jointly take as input the model scores of all documents under consideration for a given query, and can encourage the most flexible scoring configurations at the granularity of permutations rather than individual items and commonly rely on the Plackett-Luce probability model over permutations \citep{luce1959individual, plackett1975analysis}.

\vspace{-0.6em}
\paragraph{Generative Reranking}

Autoregressive decoder LLMs introduce the facet of text generation on top of the simple API of ``model-as-document-scorer''. The score-based loss formulations discussed above remain applicable insofar as a docID's probability can be obtained from the probabilities of its constituent tokens; notable LLM-based reranking methods \citep{sachan2022improving,drozdov2023parade} rely specifically on LLM-based log-likelihood to obtain scores. Meanwhile, the terminology of ``pointwise'' vs ``listwise'' naturally acquires a second meaning, i.e., whether the text string which the LLM generates comprises a single docID or a ranked list of docIDs (separated by a delimiter and in descending order, as per common convention). Our ``pointwise'' approach relates to both these senses (scoring just one document, or generating just a single docID string), since steering LLM probabilities to be more rank-aware also directly impacts the rankings produced by a decoding strategy like beam search at inference time. \textit{The generative approach naturally allows us to perform ranking with a single model call, in contrast to CEs which evaluate the score for every document separately. The compact pointwise output format mitigates snowballing errors and is efficient to train under our rank-aware modification to next-token prediction.}

A handful of examples illustrate the diverse reranking paradigms afforded by generative LLMs. RankVicuna \citep{pradeep2023rankvicuna} and RankZephyr \citep{pradeep2023rankzephyr} explore training/distillation choices, using instruction tuning with next token prediction, to achieve a robust open-source zero-shot LLM for listwise ranking. LTRGR \citep{li2024learning}, on top of a next-token-prediction-based ``learning-to-generate'' phase, introduces a ``learning-to-rank'' phase that maximizes the margin between the autoregressively-computed scores of relevant and irrelevant passage identifiers. RecRanker \citep{luo2024recranker} trains an LLM with instruction tuning on a combination of pointwise, pairwise, and listwise ranking objectives for top-$k$ sequential recommendation, and at inference time aggregates the three types of responses into a hybrid ranking by means of respective utility functions. GenRank \citep{huang2025towards} identifies the autoregressive generation mechanism as critical to the effectiveness of generative ranking, more so than the pre-training paradigm. \cite{sun2023chatgpt} explore the ranking capabilities of ChatGPT and GPT-4 and propose a distillation technique to learn smaller specialized rankers. \citet{gupta2025scalable} introduce BlockRank, efficiently tailoring the attention mechanism to in-context ranking. \textit{We argue that autoregressive methods which employ next-token prediction or instruction fine-tuning are relying on a token-level loss which is fundamentally not rank-aware. Our proposed loss function is a simple, principled and effective injection of rank-awareness into next-token prediction — a loss function which is ubiquitous to autoregressive models.}

\vspace{-0.6em}
\paragraph{Parametric Generative Retrieval and DocID Creation} Differentiable Search Index (DSI) \citep{tay2022transformer} instills the docID index in model weights as part of the training objective. A critical component of parametric generative retrieval is effective docID design \citep{li2023multiview} — we take up this question for the ESCI dataset in \Cref{sec:esci-dataset}. Note that we opt to use in-context ranking \citep{gupta2025scalable}, which leverages the LLM's in-context reasoning and copying abilities to perform retrieval and ranking from a pool of documents, with the docIDs provided as part of the prompt \cite{wang2024large}. For in-context ranking, one can interpret the longer prompt as an ``augmented query''; our theory (\Cref{sec:capacity-for-ranking}) and loss function (\Cref{sec:rank-aware-training-loss}), however, generalizes to parametric as well as in-context ranking.

\section{The Capacity for Ranking}\label{sec:capacity-for-ranking}

Let $\mathcal{Q}$ be a set of queries and $\mathcal{D}$ be a corpus of documents (or their docIDs).
Towards establishing a theory of the ARR's model capacity for ranking, we start with a formal definition of the ranking task as follows.
\begin{definition}[Ranking Task]
    We consider a \textbf{ranking task} where for each query $q \in \mathcal{Q}$, the goal is to generate a list $\mathscr{L}(q)  = [d_1(q), \ldots, d_k(q)]$ of relevant documents, $d_i(q) \in \mathcal{D}$. Ordering is significant: $\forall i < j$, $d_i(q)$ is considered more relevant to $q$ than $d_j(q)$, denoted by $d_i(q) \succ d_j(q)$.
\end{definition}

We refer to \textbf{Autoregressive Ranking (ARR)} as the use of a causal language model for ranking tasks.  Such a model predicts the probability of a sequence of tokens conditioned on its preceding context. Given a query $q$, an ARR model defines a conditional probability distribution over document token sequences. The probability of any document $d \in \mathcal{D}$ conditioned on a query $q \in \mathcal{Q}$ is the product of $d$'s token probabilities, calculated autoregressively.\footnote{Fortunately, the autoregressive nature of causal models offers a more efficient solution: Instead of scoring, we can directly generate the high-probability documents using greedy decoding or beam search. By employing beam search, we can generate the most likely documents conditioned on the query, directly yielding a ranked list without the need for exhaustive computation.} 

Both ARR and dual encoders (DEs) are examples of a \textbf{scoring-based ranking architecture}. Such architectures generate a score for each document given a query $q$, with the objective that the scores, when sorted in descending order, align with the ground-truth ordering of the relevant documents, and that scores for irrelevant documents are smaller than scores for relevant ones. We state this formally below.
\begin{definition}[Scoring-based Ranking Architecture]
    A (parametric) \textbf{scoring-based ranking architecture} is a function $f(q, d; \theta): \mathcal{Q} \times \mathcal{D} \rightarrow \mathbb{R}$, which takes a pair $q, d$ and produces a scalar score interpreted as the similarity. Here, $\theta$ denotes learnable model parameters.
\end{definition}

\noindent In particular, we say that a scoring-based architecture $f(q, d; \theta)$ \textbf{can solve a ranking task} if there exists a $\theta$ such that the scores solve the ranking task for all queries in $\mathcal{Q}$.

To simplify our analysis, we consider an idealized case of the ranking task defined as follows.
\begin{definition}[Complete Ranking Task]
    Given $\mathcal{Q}$ and $\mathcal{D}$, we call a ranking task \textbf{complete} if:
    \begin{itemize}[align=left,leftmargin=*,itemsep=0pt,topsep=0pt]
        \item $\forall q \in \mathcal{Q}$, all documents are relevant, i.e., $\mathcal{D}(q) = \mathcal{D}$
        \item $\forall$ ordering $\pi(\mathcal{D})$ of the documents $\mathcal{D}$, $\exists q \in \mathcal{Q}$ s.t. its ranked list of relevant documents $\mathscr{L}(q) \equiv \pi(\mathcal{D})$.
    \end{itemize}
\end{definition}
\noindent We proceed to show that the architectural capacity of ARR for ranking is not limited in the same way as that of a DE.

\subsection{Ranking via Dual Encoders}

DEs are a popular choice for information retrieval owing to their computational efficiency.
Unfortunately, DEs are usually insufficient due to their limited capacity compared to cross encoders (CEs), which are often critically employed in a subsequent reranking stage. 
Rigorous analysis of the capacity of DEs has been performed in the case of \emph{retrieval}, where the ranking of the relevant documents is not considered. 
In particular, \cite{guo2019breaking} shows that an embedding space dimension that grows linearly with the number of relevant documents is \emph{sufficient} for a DE to correctly separate relevant documents from irrelevant ones. 
While this result hints at the potential difficulties of DEs in handling many relevant documents, it remains unclear whether such a dimension is also \emph{necessary}.

In this section, we start with an analysis of a DE's capacity for the \emph{ranking} task. 
Our result shows that an embedding dimension that grows linearly with the number of relevant documents is \emph{necessary} for a DE to correctly rank them. 

In DEs, the query and document are embedded into the same Euclidean space, with the negated Euclidean distance as the similarity score. 
Specifically, a dual encoder generates a similarity score for each pair of $q \in \mathcal{Q}$ and $d \in \mathcal{D}$ as
\begin{equation}\label{eq:dual-encoder-ranking-architecture}
    f_\text{DE}(q, d; \theta) = -\|E_Q(q; \theta_\mathcal{Q}) - E_D(d; \theta_\mathcal{D})\|_2    
\end{equation}
where $E_Q: \mathcal{Q} \to \mathbb{R}^n$ and $E_D: \mathcal{D} \to \mathbb{R}^n$ are query and document encoders, respectively, and $n$ is the dimension of the embedding space.
Here, we assume that $E_Q$ and $E_D$ can have an arbitrary architecture. This means that they can be taken to be as expressive as needed when studying whether or not DEs are expressive enough to solve the ranking task. Our main result for DEs is the following.

\begin{prop}[DE Insufficiency for Complete Ranking]
Let $f_\text{DE}(q, d; \theta)$ be a dual encoder architecture defined in \Cref{eq:dual-encoder-ranking-architecture} with some embedding dimension $n$, and let $k := |\mathcal{D}|$. Then there does not exist a $\theta$ such that $f_\text{DE}(q, d; \theta)$ solves a complete ranking task when
    \begin{equation}\label{eq:dual-encoder-dimension-bound}
        n < \frac{\ln k!}{2 \ln k}, 
    \end{equation}
\end{prop}
\begin{proof}
    Let $N_{n, p}(k)$ denote the maximum number of distinct distance permutations generated by $k$ sites in $\mathbb{R}^n$ with the $L_p$ metric. From \citet{skala2009counting} Corollary 8, $N_{n, 2}(k) \le k^{2n}$. When Inequality (\ref{eq:dual-encoder-dimension-bound}) holds, we have $k^{2n} < k^!$, and thus $n$-dimensional embeddings are unable to capture all $k!$ permutations required in a complete ranking task. 
\end{proof}

Using Stirling's formula which states $\ln(k!) = k \ln k - k + O(\ln(k))$, the condition in \Cref{eq:dual-encoder-dimension-bound} becomes $n < \frac{k}{2} - \frac{k}{2\ln k} + O(1)$. Since $k / \ln k$ grows sublinearly in $k$, \textbf{this result states that the embedding dimension $n$ of the DE needs to grow linearly with the number of documents $k$}.

Note that the result holds regardless of the choice of the particular architecture for the encoders $E_Q$ and $E_D$. 
This means that even if the encoder architectures are expressive enough to realize any function, a DE will not be able to solve a complete ranking task if $n$ is not sufficiently large.

\subsection{Ranking via (Pointwise) Autoregressive Rankers}

We turn our attention to the ability of an ARR to rank documents using docIDs.\footnote{Proofs for this section can be found in \Cref{sec:proofs}.} As is standard, we assume the ARR uses token embeddings $\operatorname{E}\in \mathbb R^{|\mathcal{V}| \times n}$ and computes next-token probabilities over the vocabulary $\mathcal{V}$ as  $P(\cdot \mid c) = \sigma(\operatorname{E} \varphi(c))$ where $\varphi(c) \in \mathbb R^n$ is a hidden vector produced by the ARR given context $c$ and $\sigma$ is the softmax.

In the setting of document ranking, documents can be represented by unique identifiers, or docIDs, where each docID is a sequence of tokens from $\Vdocids \subseteq \mathcal{V}$. To ensure the model only generates valid docID tokens we can make use of constrained decoding, which restricts next-token probabilities to those tokens in $\Vdocids$ by computing
\[
P(v\mid c) = \begin{cases}
\sigma(\Edocids \varphi(c)) \quad & \text{if} \quad v \in \Vdocids,\\
0 \quad & \text{otherwise.}
\end{cases}
\]
where $\Edocids \in \mathbb{R}^{|\Vdocids| \times n}$ represents a matrix whose rows are the token embeddings associated with $\Vdocids$. To study the rank requirements of $\Edocids$, we introduce for convenience the matrix $\Eprime \coloneqq [\Edocids \mathbf{1}_{|\Vdocids|}] \in \mathbb{R}^{|\Vdocids| \times (n+1)}$, i.e., $\Edocids$ augmented with a column of ones.

We consider an idealized ``infinite-capacity'' setting where the hidden vector $\varphi (c)$ can be any vector in $\mathbb R^n$.

\begin{restatable}[Arbitrary Distributions over Tokens]{prop}{PropArbitraryTokenLevelProb}
\label{prop:arbitrary-tok-level-prob}
An infinite-capacity ARR can produce any (strictly positive) probability distribution over tokens in $\Vdocids$ via constrained decoding iff $\operatorname{rank}(\Eprime) = |\Vdocids|$.
\end{restatable}

\noindent The contrapositive of the forward implication is commonly known as the ``softmax bottleneck'' \citep{yang2017breaking}.

\noindent We note that many existing language models come with ``special tokens'' which can be used as $\Vdocids$, and for any existing language model which does not come with pre-allocated ``special tokens'' we can expand the vocabulary with $n$ new tokens and extend its embedding matrix with the $I_{|\Vdocids|}$ identity matrix, which would satisfy the requirements. An implicit limitation of this is that $|\Vdocids| \le n$, however the chain-rule of probability implies the following.

\begin{corollary}[Arbitrary Distributions over Sequences]\label{cor:arbitrary-seq-level-prob}
An infinite-capacity ARR can produce any (strictly positive) probability distribution over sequences of tokens $\Vdocids$ via constrained decoding iff $\operatorname{rank}(\Eprime) = |\Vdocids|$.
\end{corollary}

Producing arbitrary probability distributions over sequences clearly entails the ability to produce arbitrary \textit{rankings} over sequences. One may wonder, however, whether the requirement of \Cref{prop:arbitrary-tok-level-prob} / \Cref{cor:arbitrary-seq-level-prob} is also \textit{necessary} for ranking. Surprisingly, at a single token position the requirement for ranking is the same as for distributions:

\begin{restatable}[Arbitrary Permutations over Tokens]{prop}{PropArbitraryTokenLevelRanking}
\label{prop:arbitrary-tok-level-ranking}
An infinite-capacity ARR can produce any permutation over tokens in $\Vdocids$ via constrained decoding iff $\operatorname{rank}(\Eprime) = |\Vdocids|$.
\end{restatable}

\noindent We observe that the capacity to produce arbitrary permutations over \textit{sequences} depends on the capacity to express all $|\Vdocids|!$ rankings at \textit{a single token position}. Since the requirement for this is the same as the requirement to represent arbitrary distributions at both token- and sequence-level, we conclude that the rank requirement expressed in \Cref{prop:arbitrary-tok-level-prob}, \Cref{cor:arbitrary-seq-level-prob}, and \Cref{prop:arbitrary-tok-level-ranking} is not only sufficient but necessary for the ranking task.

\textit{\textbf{Remark:} Is the theoretical requirement on $\Eprime$ satisfied by LLMs used in practice for reranking, and why are those LLMs successful even when this requirement clearly can not be satisfied? Our experiments on WordNet and ESCI (\Cref{sec:experimental-setup}) reflect two common scenarios. In WordNet, where docIDs comprise English nouns, the docID tokens span almost all of $|\mathcal{V}|$. Since the model dimension is much smaller, $\Eprime$ can not achieve a rank even close to $|\Vdocids|$.\footnote{Mistral-7B-v0.3-it has $|\mathcal{V}|=32,768$ and $d=4,096$.} However, most of the possible sequences are simply never encountered, and the model is tasked with learning a very small subset of all possible permutations. In ESCI, by contrast, we choose numerical docIDs, and the resulting $|\Vdocids|$ is therefore closer to the model dimension, allowing the rank requirement on $\Eprime$ to be matched more accurately.}

\textbf{At either the single token position or the sequence level, we see that the same requirement on the docID token embeddings} (namely, that $\operatorname{rank}(\Eprime) = |\Vdocids|$, i.e., slightly weaker than $\Edocids$ being full-rank) \textbf{is both necessary and sufficient to enable a fixed-dimension ARR to generate arbitrary probability distributions or arbitrary permutations for an unrestricted number of documents.} This ability is strictly superior to DEs, whose embedding dimension must grow linearly with the number of documents to achieve arbitrary rankings.

\section{Generalized Rank-Aware Training Loss}\label{sec:rank-aware-training-loss}

Given an example ``point'' of the form $(q, d_r, r)$ with the query, a docID, and the true rank, we propose a generalized loss function for autoregressive ranking which can be viewed as a weighted item-level loss, with the weight coefficient $\lambda(r)$ for a given docID being a function of the document's rank $r$. The item-level loss decomposes as a sum of cross entropy losses between the model predictions and target distributions \textbf{y} at each timestep (token position) of the docID under consideration. Note that the target distributions are a function of the timestep $t$ and potentially also of the current docID's rank $r$, hence $\mathbf{y}(r, t)$. Our generalized next-token prediciton loss for ranking, \textsc{SToICaL} (for ``\underline{S}imple \underline{To}ken-\underline{I}tem \underline{Ca}librated \underline{L}oss''), is given below. $\theta$ represents model parameters, $\mathcal{T}$ the tokenizer, and $\bar{q}$ the ``augmented query'' (cf. \Cref{sec:experimental-setup}):
\begin{equation}\label{eq:general-loss-formulation}
\mathcal{L}(q, d_r, r; \theta) = \lambda(r) \sum_{t=1}^{|\mathcal{T}(d_r)|} \biggl(\operatorname{CE}\Bigl(\mathbf{y}(r, t), p_\theta(\mathcal{T}(d_r)[t] \mid \mathcal{T}(\bar{q}), \mathcal{T}(d_r)_{<t})\Bigr) \biggr)
\end{equation}
\noindent Observe that if $\mathbf{y}(r,t)$ is set to a one-hot vector, the summation across a docID's timesteps simplifies to the negative log-likelihood of that docID. More generally, we can view our loss as a weighted sum of item-level log-likelihoods, for some reweighting function $\lambda(r)$, and for some customizable construction $\mathbf{y}(r,t)$ of target distributions per timestep.

The most common target distribution $\mathbf{y}(r, t)$ is one-hot:
\begin{equation}
\mathbf{y}(t)_i \coloneqq \begin{cases}
    1 & \text{if } v_t = i \\
    0 & \text{otherwise}
\end{cases}
\end{equation}
\noindent \Cref{sec:prefix-tree} presents a rank-aware improvement to $\mathbf{y}(r,t)$.

\subsection{Rank-Aware Item-Level Reweighting}\label{sec:rank-aware-item-level-reweighting}

A simple improvement to next-token prediction is to reweight the loss contribution of $d_r, r \in [1, n_q]$ by setting $\lambda(r)$ to a function that decreases as rank $r$ increases (rank 1 represents the most relevant document). We consider two simple variants: a \textit{fractional} reweighting function $\lambda(r) = \frac{1}{r^\alpha}$ where $\alpha$ is a temperature hyperparameter, or a \textit{stepwise} reweighting function $\lambda(r) = \frac{n_q - r + 1}{n_q}$ where $n_q$ is the total number of documents for the query. If we decide to omit some documents from the training (for example, train a top-1 retrieval model only) we can still express this by setting our reweighting function to the \textit{indicator} function ($\lambda(r) = \mathbb{I}_{r=1}$ in the top-1 case).

\subsection{Prefix Tree for Rank-Aware Token-Level $\mathbf{y}(r,t)$}\label{sec:prefix-tree}

A prefix tree (aka trie) groups strings according to their shared prefixes for quickly determining (by traversing left-to-right, i.e., from the root down) whether a given sequence belongs to the valid set of sequences or not. Building a prefix tree over the list of tokenized docIDs allows us to examine, at each timestep, which continuation docIDs are permissible given the partially-generated docID, and to distribute mass accordingly among valid tokens in the target distribution.

Given tokenized target docIDs in the following toy example:
\begin{itemize}[noitemsep, topsep=0pt] 
    \item $d_1$ : ``\texttt{dog}'' $\Rightarrow$ ([\texttt{d}, \texttt{og}])
    \item $d_2$ : ``\texttt{cat}'' $\Rightarrow$ ([\texttt{c}, \texttt{at}])
    \item $d_3$ : ``\texttt{cats}'' $\Rightarrow$ ([\texttt{c}, \texttt{at}, \texttt{s}])
    \item $d_4$ : ``\texttt{deer}'' $\Rightarrow$ ([\texttt{d}, \texttt{eer}])
    \item $d_5$ : ``\texttt{fish}'' $\Rightarrow$ ([\texttt{f}, \texttt{ish}])
\end{itemize}
\vspace{-0.5em}
\noindent we can construct a prefix tree over their tokens, as shown in \Cref{fig:example}(a). The leaf nodes and potentially some nonterminal nodes of the trie, which we term \textit{docID nodes} and represent with solid borders in \Cref{fig:example}(a), will each represent the tokenization of a docID $d_r$, and to such nodes we assign the score $\mu(r) = \frac{1}{r^\beta}$ for temperature hyperparameter $\beta$. To score any given node in the trie (e.g., for a partially-generated docID), we can marginalize over the scores of all docIDs which that node governs.\footnote{At marginalization, the fractional $\mu(r)$ prevents a continuation which leads to the top-$r$ docID from being outweighed by a continuation leading to multiple lower-ranked docIDs. As $\beta \rightarrow \infty$, the resulting distributions converge to one-hots.}

\begin{figure}[h] %
    \centering %
    \includegraphics[width=\linewidth]{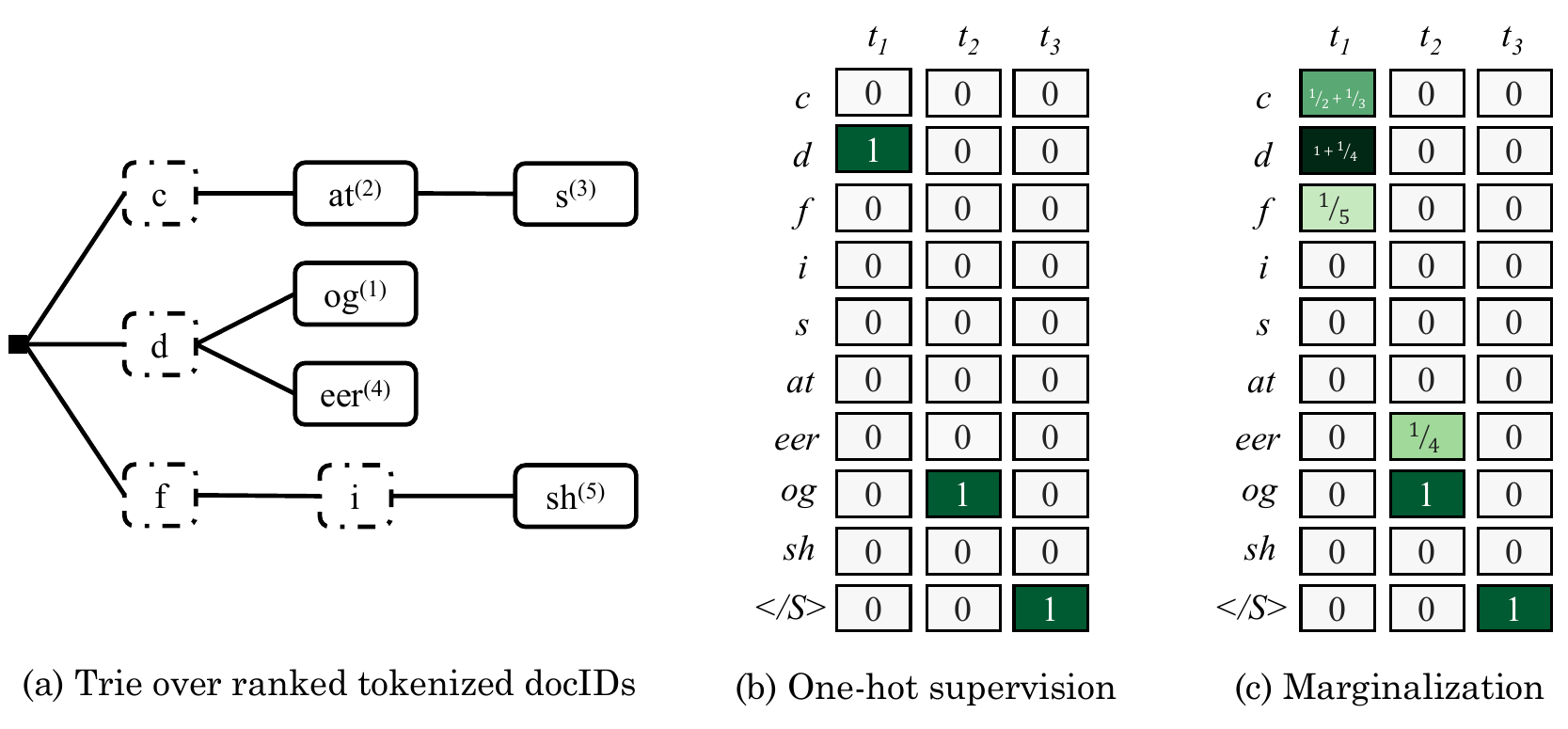}
    \caption{Illustration of (a) a prefix-tree; target distributions derived from it via (b) one-hot supervision versus (c) marginalization.}
    \label{fig:example}
\end{figure}

Given an empty string ($\varepsilon$) at timestep $t_1$, the trie tells us that permissible continuation tokens for [$\varepsilon$] are $\{\texttt{c, d, f}\}$. At timestep $t_2$, if our highest-ranked document is $d_1 = \texttt{``dog''}$, we will have teacher-forced token \texttt{d} as the $t_2$ input to the model, and the prefix tree will consequently yield tokens $\{\texttt{og, eer}\}$ as valid continuations for [$\varepsilon$, \texttt{d}]. We illustrate the (pre-normalized) target distributions produced by this trie-based marginalization process in \Cref{fig:example}(c).\footnote{Technically, for timesteps $t>1$, teacher forcing causes a mismatch between the rank-aware training described here and the model's actual decision-making at inference time. For example, if the model chooses \texttt{``c''} at timestep \textit{t}=1 during inference, we will not have smoothed over \texttt{``at''} at $t=2$ during training to encourage the second-best \texttt{``cat''} to be generated (without constrained decoding, we might well generate \texttt{``c''} $+$ \texttt{``og''} $\rightarrow$ \texttt{``cog''} in such a scenario). Nevertheless, akin to \citet{reddy2024first}, our method does encourage full rank-awareness at $t=1$. We leave the question of mitigating this rank-aware training-inference mismatch for $t>1$ to future work.}

\section{Experiments}

\vspace{0.1cm}
\subsection{Datasets, Experimental Settings, and Metrics}\label{sec:experimental-setup}

\paragraph{WordNet}\label{sec:wordnet-dataset}

WordNet \citep{fellbaum1998wordnet} includes a hierarchy over noun concepts, aka ``synsets'', where the top-level node is \texttt{entity.n.01} (the ``\texttt{n}'' suffix means \textit{noun}, and ``\texttt{01}'' that this is the first of multiple possible word senses). To obtain a query, we randomly sample a noun synset at depth one or greater (i.e., excluding the root \texttt{entity.n.01} itself). The relevant documents for the sampled query synset are its hypernym synsets, meaning all the more abstract synsets on the path from the query up to the root. For each query, we additionally sample a negative document not on the hypernym path to be scored alongside the relevant documents. For example (we omit the suffixes for brevity):
\begin{itemize}[align=left,leftmargin=*,itemsep=0pt,topsep=0pt]
    \item \textbf{Query:} deer
    \item \textbf{Ranked docIDs:} \text{ruminant} $\succ$ \text{even-toed ungulate} $\succ$ \text{ungulate} $\succ$ \text{placental} $\succ$ \text{mammal} $\succ$ \text{vertebrate} $\succ$ \text{chordate} $\succ$ \text{animal} $\succ$ \text{organism} $\succ$ \text{living thing} $\succ$ \text{whole} $\succ$ \text{object} $\succ$ \text{physical entity} \textcolor{gray}{[$\succ$ \text{entity}]}
    \item \textbf{Negative docID:} \text{solarization}
\end{itemize}
\noindent We provide a zero-shot prompt for WordNet in \Cref{sec:wn-prompt}.

\vspace{-0.6em}
\paragraph{ESCI Shopping Queries}\label{sec:esci-dataset}

We process the \textit{Shopping Queries} dataset of \cite{reddy2022shopping} as follows to obtain a reranking dataset. We use the Gecko \citep{lee2024gecko} embeddings (specifically from the \texttt{gecko-1b-en} model) of the queries and product titles to obtain a matrix of inner products between all queries and all product titles. To aggregate the relevant product titles for a query, we take the $1^\text{st}, 100^\text{th}, 200^\text{th}, \ldots, 10000^\text{th}$ product titles that have largest dot product with the query, in descending order. We choose step size 100 to ensure that the difference in relevance between the $n^\text{th}$ and the $(n+1)^\text{st}$ product title is meaningful. Finally, we filter out all non-English queries and product titles, and preserve only those examples that have a non-empty set of relevant English product titles.

Unlike the WordNet dataset (\Cref{sec:wordnet-dataset}), the documents of ESCI are (\textit{a}) substantially longer, i.e., on the order of dozens of tokens instead of ten, and (\textit{b}) and have not been sampled from a taxonomy. Issue (\textit{a}) motivates the creation of docIDs that are both short enough to generate as pointwise targets, and structured enough so that a prefix tree can be built over them for target training distributions as discussed in \Cref{sec:prefix-tree}. We create such docIDs via sparse dictionary learning over the Gecko product title embeddings to obtain sparse vectors, each of dimension 100 with 3 nonzero values.\footnote{We use Orthogonal Matching Pursuit from \texttt{scikit-learn} \citep{pedregosa2011scikit}, \texttt{n\_components=100} and \texttt{transform\_n\_nonzero\_coefs=3}.} The docID for a given product title is then taken to be the concatenation of the indices of the nonzero entries in its corresponding sparse vector, sorted in descending order by the absolute values of the entries (e.g., ``\texttt{25,36,39}''). A zero-shot prompt for ESCI is provided in \Cref{sec:esci-prompt}. 

\vspace{-0.6em}
\paragraph{Task Setup for Autoregressive Ranking}\label{sec:task-definition}

The in-context input to our Causal Transformer is a query $q$ and an unordered docID set $\operatorname{shuffle}(\{d_1, \ldots, d_{n_q}\})$; assume the ground truth ranking is $d_1 \succ \ldots \succ d_{n_q}$. We construct the prompt, denoted by $\bar{q}$ for the query augmented with the corpus and system prompt, as $\bar{q} \coloneqq \operatorname{prompt}(q, \operatorname{shuffle}(\{d_1, \ldots, d_{n_q}\}))$, cf. \Cref{sec:appendix-prompt}. We emphasize that this augmented prompt construction reflects an \textit{in-context} approach to ranking \citep{gupta2025scalable} in which the full corpus is provided as context — while this contrasts with DEs' and CEs' pairwise scoring approach, it is natural to ARR, which aims to exploit the in-context copying and reasoning abilities of decoder LLMs. For our evaluation, as a simple proxy to beam search, we perform greedy decoding for $\max_{i \in [1, n_q]} |\mathcal{T}(d_i)|$ steps,  where $\mathcal{T}$ is the tokenizer, and score any document by averaging the log-probabilities of its constituent tokens, i.e., $\operatorname{score}(d_j) = \frac{1}{|\mathcal{T}(d_j)|}\sum_{k=1}^{|\mathcal{T}(d_j)|}\log p_\theta (\mathcal{T}(d_j)[k] \mid \mathcal{T}(\bar{q}), \mathcal{T}(d_j)_{<k})$.\footnote{We emphasize that in practice, a $\operatorname{top}-k$ ranking can still be directly generated via beam search using beam size $k$.} We use Mistral-7B-v0.3-it \citep{jiang2024identifying} for all fine-tuning experiments in this paper, with settings for the generative approach in \Cref{sec:hparams-generative}.

\paragraph{Metrics}\label{sec:metrics}
We adopt the following metrics – the first one a heuristic to assess the constraint satisfaction of our approach, the other two commonly-used metrics for ranking:

\begin{itemize}[align=left,leftmargin=*,itemsep=0pt,topsep=0pt]
    \item For each example of the WordNet dataset, we use the sampled negative document to assess the \textbf{constraint violation rate (CVR)}, a heuristic measurement we define as follows. Given $n$ positive documents $d_1^{+}, \ldots, d_n^{+}$ and 1 negative (irrelevant) document $d^{-}$, we say that the model has incurred a violation if the log-probability the model assigns to any of the $d^{+}_i$ is lower than the log-probability it assigns to $d^{-}$. 
    \item We measure the \textbf{normalized discounted cumulative gain (nDCG)} between ``gold'' scores and predicted log-probabilities (including for the sampled negative document). The ``gold'' are in log-scale in decreasing order of relevance: $\operatorname{score}(d_1^{+})=\log(n_q+1), \operatorname{score}(d_2^{+})=\log(n_q) \ldots, \operatorname{score}(d_{n_{q}-1}^{+})=\log(3), \operatorname{score}(d_{n_q}^{+})=\log(2), \operatorname{score}(d_{q}^{-})=\log(1)=0$.
    \item We use the log-probabilities to $\operatorname{argsort}$ over the docIDs, and use the gold ranking $\pi$ and predicted ranking $\hat\pi$ to measure \textbf{recall-at-k (R@k)} as $\frac{|\pi[:k] \cap \hat\pi[:k]|}{k}$.
\end{itemize}

\begin{table}[ht] 
\centering

\resizebox{\textwidth}{!}{%
    \setlength{\tabcolsep}{5pt}
    \renewcommand{\arraystretch}{0.95}
    \begin{tabular}{c c ccccccc}
    \toprule
    & & CVR ($\downarrow$) & nDCG ($\uparrow$) & R@1 ($\uparrow$) & R@2 ($\uparrow$) & R@3 ($\uparrow$) & R@4 ($\uparrow$) & R@5 ($\uparrow$) \\
    \midrule
    \multicolumn{2}{c}{\parbox{3.5cm}{\centering NTP: $\lambda(r) = \mathbb{I}_{r=1}$\\$\mathbf{y}$ one-hot}} & 27.66\% & 94.89 & \textbf{99.96} & 62.43 & 55.3 & 55.16 & 63.42 \\
    \midrule
    \multirow{5}{*}{\parbox{3.2cm}{\centering $\lambda(r) = \frac{1}{r^\alpha}$\\$\mathbf{y}$ one-hot}}
    & $\alpha=1$ & 0.0\% & 99.6 & 91.1 & 87.68 & 88.81 & 92.13 & 93.94 \\
    & $\alpha=2$ & 0.0\% & \textbf{99.83} & 97.74 & 95.69 & 95.48 & \textbf{96.89} & \textbf{96.58} \\
    & $\alpha=3$ & 0.0\% & 99.81 & 99.22 & 97.35 & \textbf{96.16} & 96.07 & 95.53 \\
    & $\alpha=4$ & 0.02\% & 99.78 & 99.36 & \textbf{97.46} & 95.63 & 96.03 & 95.31 \\
    & $\alpha=5$ & 0.04\% & 99.67 & 99.6 & 97.14 & 94.55 & 95.71 & 93.52 \\
    \midrule
    \multicolumn{2}{c}{\parbox{3.5cm}{\centering $\lambda(r) = \frac{n_q - r + 1}{n_q}$\\$\mathbf{y}$ one-hot}} & 0.0\% & 99.6 & 51.62 & 69.88 & 82.09 & 89.6 & 93.34 \\
    \midrule
    \multirow{5}{*}{\parbox{3.2cm}{\centering $\lambda(r) = \mathbb{I}_{r=1}$ \\ $\mathbf{y}$ marg. over trie w/ leaf scores $\frac{1}{r^\beta}$}}
    & $\beta=1$ & 1.48\% & 96.54 & 98.1 & 72.93 & 63.06 & 59.64 & 64.25 \\
    & $\beta=2$ & 1.38\% & 96.54 & 99.3 & 67.31 & 57.51 & 56.98 & 62.45 \\
    & $\beta=3$ & 2.04\% & 96.34 & 99.44 & 67.58 & 59.95 & 58.58 & 63.17 \\
    & $\beta=4$ & 4.54\% & 96.36 & 99.58 & 68.53 & 62.25 & 62.12 & 67.41 \\
    & $\beta=5$ & 13.04\% & 96.47 & 99.78 & 70.84 & 65.12 & 64.81 & 70.43 \\
    \bottomrule
    \end{tabular}%
}
\caption{WordNet results (evaluated over 5000 (query, targets)-examples)}
\label{tab:wordnet-results}

\vspace{0.5cm} 

\resizebox{\textwidth}{!}{%
    \setlength{\tabcolsep}{5pt}
    \renewcommand{\arraystretch}{0.95}
    \begin{tabular}{c c ccccccc}
    \toprule
    & & nDCG ($\uparrow$) & R@1 ($\uparrow$) & R@2 ($\uparrow$) & R@5 ($\uparrow$) & R@10 ($\uparrow$) & R@25 ($\uparrow$) & R@50 ($\uparrow$) \\
    \midrule
    \multicolumn{2}{c}{\parbox{3.5cm}{\centering NTP: $\lambda(r) = \mathbb{I}_{r=1}$\\$\mathbf{y}$ one-hot}} & 95.23 & \textbf{95.16} & 52.58 & 27.64 & 23.51 & 37.98 & 62.99 \\
    \midrule
    \multirow{5}{*}{\parbox{3.5cm}{\centering $\lambda(r) = \mathbb{I}_{r=1}$ \\$\mathbf{y}$ marg. over trie w/ leaf scores $\frac{1}{r^\beta}$}}
    & $\beta = 1$ & \textbf{97.21} & 70.0  & 56.61 & 45.57 & 46.27 & 54.78 & \textbf{69.58} \\
    & $\beta = 2$ & \textbf{97.21} & 70.32 & \textbf{58.06} & \textbf{51.07} & \textbf{48.08} & \textbf{54.96} & 69.03 \\
    & $\beta = 3$ & 97.11 & 68.71 & 56.13 & 48.03 & 45.78 & 52.31 & 67.63 \\
    & $\beta = 4$ & 96.96 & 69.03 & 54.84 & 44.6  & 43.9  & 50.29 & 67.7  \\
    & $\beta = 5$ & 96.89 & 66.77 & 56.13 & 44.21 & 42.4  & 50.72 & 67.61 \\
    \bottomrule
    \end{tabular}%
}
\caption{ESCI results (evaluated over 310 (query, targets)-examples)}
\label{tab:esci-results}

\end{table}

\subsection{Main Results}

In Tables \ref{tab:wordnet-results} and \ref{tab:esci-results} we present our experimental results on the WordNet and ESCI datasets, respectively. Each of our \textbf{ranking-aware loss functions can be seen to drastically reduce the constraint violation rate} (cf. \Cref{tab:wordnet-results}) — i.e., the model learns to score irrelevant docIDs lower than relevant docIDs for a query. Such a shift in probability mass from invalid to valid docIDs makes the resulting LLM more aligned, as it were ``by construction'', with the constraints of generating valid docIDs. Such realignment would naturally ease the computational burden of a debiasing algorithm for constrained decoding, e.g., as described in \citet{ye2025efficient}.

Fixing our \textbf{y} to one-hots, we observe for WordNet in \Cref{tab:wordnet-results} that \textit{fractional} ($\lambda(r) = \frac{1}{r^\alpha}$) item-level reweighting is more effective than \textit{stepwise} ($\lambda(r) = \frac{n_q - r + 1}{n_q}$), while in the limit $\alpha \rightarrow \infty$ for fractional reweighting, the target distributions look more and more like a top-1 training curriculum $\frac{1}{r^\alpha} \rightarrow \mathbb{I}_{r=1}$ (cf. \Cref{sec:rank-aware-item-level-reweighting}). We separately explore the effect of rank-aware target distributions $\mathbf{y}(r,t)$ (cf. \Cref{sec:prefix-tree}), with a single rank-aware target per query (rather than per document $d_r$). We similarly observe decreased constraint violation rate, improved nDCG and R@K for $K > 2$, but not as dramatic as for the item-level reweighting loss.

The ESCI dataset contains a much higher number of product titles to be reranked per query on average than the WordNet dataset (cf. \Cref{sec:dataset-statistics}). Training with item-level losses until convergence would treat each docID $d_r$ for $r \in [1, n_q]$ as a separate target, and would thus require a proportionally larger number of training steps to iterate over the entire dataset. Hence, we run our ESCI experiments under the same budget constraint as next-token prediction, providing a single target per every query by setting $\lambda(r) = \mathbb{I}_{r=1}$, and inject ranking information via trie-based marginalization. While hurting R@1 performance, \Cref{tab:esci-results} improves R@K for all $K>1$ as well as the aggregate nDCG score. \textbf{Trie-based marginalization can thus be viewed as an economical alternative to the item-level reweighting}, even if it performs second to item-level reweighting under the assumption of unlimited training examples. Additionally, the improved ESCI ranking scores highlight the \textbf{synergy between this token-level loss and the ESCI docIDs, which we specifically learned to be trie-friendly}. We hope this result stimulates future research into optimal combinations of docID learning algorithms with token-based ranking losses.

\subsection{Comparison with DEs and CEs}\label{sec:comparison-with-de-and-ce}
\vspace{0.2cm}

We compare (pointwise) autoregressive ranking with DEs and CEs. We conduct our experiment on WordNet, and report the results in \Cref{fig:wordnet-de-ce}(a). It can be seen that autoregressive ranking is on par with CEs while being significantly more powerful than DEs. Below we provide details on our choice of DEs and CEs, and describe the hyperparameter settings in \Cref{sec:hparams-de} and \Cref{sec:hparams-ce}, respectively.

Since the number of noun synsets in WordNet is relatively small (i.e., 82,155), we simply use an embedding table of size 82,155$\times n$ to map each synset to an $n$-dimensional embedding (as opposed to using a tokenizer). 
Denote this embedding lookup function as $\operatorname{E}_\phi(\cdot)$, where $\phi$ is the trainable embedding table. 
Given a pair $(q, d)$, we use DEs with the inner product $\langle \operatorname{E}_\phi(q), \operatorname{E}_\phi(d) \rangle$ as their similarity score. In other words, we use lookup table DEs which is a common choice on this dataset \citep{dhingra2018embedding,parmar2022hyperbox,an2023coarse}.
Given a batch of positive data $\{(q_i, d_i, r_i)\}_{i=1}^B$, we train the DE by minimizing the following \textbf{weighted} batch softmax loss:
\begin{equation}
    \mathcal{L}_\text{DE}(\{q_i, d_i, r_i\}_{i=1}^B; \phi) = -\frac{1}{B} \sum_{i=1}^B 
    \left(
    \lambda(r_i) \cdot \log
    \frac{
        \exp{
            \Big(
                \langle \operatorname{E}_\phi(q_i), \operatorname{E}_\phi(d_i) \rangle / \tau
            \Big)
        }
    }
    {
        \sum_{j=1}^B 
        \exp{
            \Big(
                \langle \operatorname{E}_\phi(q_i), \operatorname{E}_\phi(d_j) \rangle / \tau
            \Big)
        }
    }
    \right).
\end{equation}
where $\tau$ is a scalar that we fix to be 0.05. Note that this formulation incorporates the reweighting factor $\lambda(r)$.
If $\lambda(r) = 1$ then the loss reduces to the standard batch softmax loss commonly used in information retrieval \citep{wu2024effectiveness}. 
Here, we follow the same idea as in \eqref{eq:general-loss-formulation} and use $\lambda(r)=\frac{1}{r^\alpha}$. 
We report the results for varying $\alpha$ in \Cref{fig:wordnet-de-ce}(b), which shows that an increasing $\alpha$ improves Recall$@k$, particularly for small $k$. 
This result shows that our reweighting factor can be useful broadly for ranking problems, beyond the particular autoregressive ranking approach. 

For CE, we concatenate the embeddings into a $2n$-dimensional vector which is passed to an MLP $f_\text{MLP}(\operatorname{concat}(\operatorname{E}_\phi(q), \operatorname{E}_\phi(d)); \psi): \mathbb{R}^{2n} \mapsto \mathbb{R}$, where $\psi$ is trainable parameters. We use ReLU activation function, three hidden activations of size $2n$, and a final output of dimension $1$ as the similarity. To train a CE, one cannot use the same batch softmax loss as for the DEs since scoring each query against all in-batch negatives is costly. Here, for a batch of positive data $\mathcal{P} = \{(q_i, d_i, r_i)\}_{i=1}^B$ we create a negative set of the same size given by $\mathcal{N} = \{(q_i, d_{\rho(i)})\}_{i=1}^B$, where $\rho: \{1, \ldots, B\} \mapsto \{1, \ldots, B\}$ is a is a randomly-chosen permutation.  Then we minimize:
\begin{equation}
\begin{split}
    \mathcal{L}_\text{CE}(\mathcal{P}, \mathcal{N}; \phi, \psi) = - \sum_{i=1}^B 
    \Bigg(
    & \frac{\lambda(r_i)}{\sum_{j=1}^B \lambda(r_j)} \log \sigma
    \Big( 
    f_\text{MLP}(\operatorname{concat}(\operatorname{E}_\phi(q_i), \operatorname{E}_\phi(d_i)); \psi)
    \Big)
    \\
    & +
    \frac{1}{B}
    \log \sigma
    \Big(
    -f_\text{MLP}(\operatorname{concat}(\operatorname{E}_\phi(q_i), \operatorname{E}_\phi(d_{\rho(i)})); \psi)
    \Big)
    \Bigg),
\end{split}
\end{equation}
\noindent with $\sigma(\cdot)$ the sigmoid function. 
A reweighting term $\lambda(r_i)$ is again added to positive pairs (effect illustrated in \Cref{fig:wordnet-de-ce}(c)).

\begin{figure}[h] %
    \centering %
    \includegraphics[width=\linewidth]{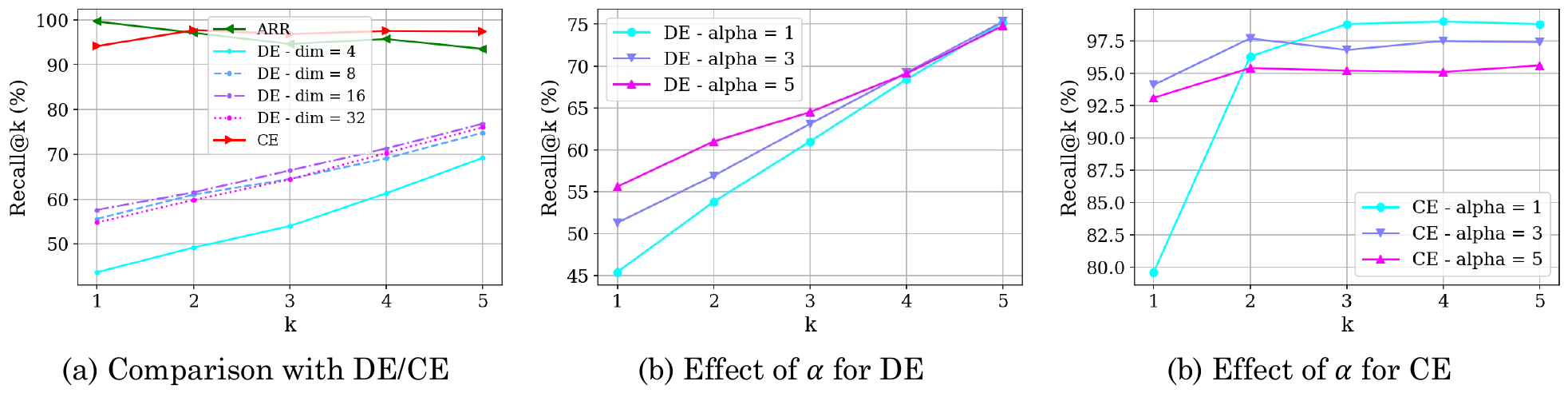}
    \caption{Results on WordNet with Dual Encoders (DEs) and Cross Encoders (CEs).
    (a) Comparison of autoregressive ranking with DEs ($n \in \{4,8,16,32\}, \alpha=5$) and CE ($n=32$). 
    (b, c) Effect of $\alpha$ in the reweighting function for DEs ($n=8$) and CEs ($n=32$), respectively.}
    \label{fig:wordnet-de-ce}
\end{figure}

\section{Conclusion}\label{sec:conclusion}
\vspace{0.1cm}
In this paper, we established a theoretical foundation for Autoregressive Ranking, proving that while DEs require embedding dimensions to grow with corpus size, ARR models generating multi-token docIDs can solve complete ranking tasks with a constant hidden dimension, given a mild condition on the rank of the embedding matrix for the docID tokens. We proposed a generalized rank-aware training loss for (pointwise) autoregressive ranking that relies on item-level reweighting and prefix-tree marginalization to distribute probability mass over valid
docID tokens based on their ground-truth relevance. Experiments on WordNet and ESCI show this approach successfully suppresses invalid docID generations and improves on key ranking metrics.

\section{Acknowledgements}\label{sec:acknowledgements}

The authors would like to thank Aditya K Menon for valuable discussions about this paper and generative information retrieval more broadly.

\bibliographystyle{abbrvnat}
\bibliography{references}


\appendix

\newpage
\section{Proofs}\label[appendix]{sec:proofs}

\subsection{Proof of \Cref{prop:arbitrary-tok-level-prob}}

\PropArbitraryTokenLevelProb*

\begin{proof}
($\Rightarrow$) Assume, \textit{reductio ad absurdum}, that $\operatorname{rank}(\Eprime) < |\Vdocids|$. This implies that there exists a vector $h \in \mathbb{R}^{|\Vdocids|}, h \neq \mathbf{0}_{|\Vdocids|}$, that is orthogonal to the span of $\Eprime$, i.e., for any $\varphi^\prime \in \mathbb{R}^{n+1}$: 

\begin{equation}\label{eq:dot-product-0-original}
h^T(\Eprime \varphi^\prime) = \sum_{i=1}^{|\Vdocids|}h_i(\Eprime \varphi^\prime)_i = 0.
\end{equation}

\noindent Let $\mathcal P \coloneqq \{i \mid h_i > 0\}$ and $\mathcal N \coloneqq \{i \mid h_i < 0\}$. Then we can write

\begin{equation}\label{eq:dot-product-0-decomposed}
\sum_{i \in \mathcal P}h_i(\Eprime \varphi^\prime)_i + \sum_{i \in \mathcal N}h_i(\Eprime \varphi^\prime)_i = 0
\end{equation}

\noindent from which it follows that

\begin{equation}\label{eq:dot-product-equality}
\sum_{i \in \mathcal{P}}h_i(\Eprime \varphi^\prime)_i = \sum_{i \in \mathcal N}|h_i|(\Eprime \varphi^\prime)_i,
\end{equation}

\noindent and since $\mathbf{1}_{|\Vdocids|}$ is clearly in the span of $\Eprime$, we also have

\begin{equation}\label{eq:dot-product-equality-S}
S \coloneqq \sum_{i \in \mathcal{P}}h_i = \sum_{i \in \mathcal{N}}|h_i|.
\end{equation}

Since our ARR can produce any strictly positive probability distribution over $\Vdocids$, it must be able to produce hidden vector $\varphi(c)$ corresponding to logits $z = \Edocids\varphi(c)$ with $z_i > z_j \text{ for all } i \in \mathcal N, j \in \mathcal P$ — such logits correspond to a probability distribution in which all token indices for which $h$ takes negative values are assigned greater probability mass than any of the token indices for which $h$ takes positive values.

We can choose $\varphi^\prime \in \mathbb{R}^{n+1}$ such that $\varphi^\prime_{1:n} = \varphi(c)$, and $z^\prime = \Eprime \varphi^\prime = \Edocids\varphi(c) + \varphi^\prime_{n+1}\mathbf{1}$ will have the same relation between values indexed by $\mathcal N$ and by $\mathcal P$ as does $z$, i.e.,

\begin{equation}\label{eq:a-b-inequality}
a \coloneqq  \max \{z^\prime_i \mid i \in \mathcal{P}\} < \min \{z^\prime_i \mid i \in \mathcal{N}\} \eqqcolon b.
\end{equation}

\noindent Recall from \Cref{eq:dot-product-equality} that $\sum_{i \in \mathcal{P}}h_i z^\prime_i = \sum_{i \in \mathcal N}|h_i| z^\prime_i$. Now,

\begin{equation}\label{eq:inequality-aS}
\sum_{i \in \mathcal{P}}h_i z^\prime_i \overset{(\ref{eq:a-b-inequality})}{\leq} \sum_{i \in \mathcal{P}}h_i a = a\sum_{i \in \mathcal{P}}h_i \overset{(\ref{eq:dot-product-equality-S})}{=} aS,
\end{equation}

\noindent and

\begin{equation}\label{eq:inequality-bS}
\sum_{i \in \mathcal{N}}|h_i| z^\prime_i \overset{(\ref{eq:a-b-inequality})}{\geq} \sum_{i \in \mathcal{N}}|h_i| b = b\sum_{i \in \mathcal{N}}|h_i| \overset{(\ref{eq:dot-product-equality-S})}{=} bS.
\end{equation}

\noindent Since $aS < bS$, we get $\sum_{i \in \mathcal{P}}h_i z^\prime_i < \sum_{i \in \mathcal{N}}|h_i| z^\prime_i$, which is a contradiction ($\bot$) with \Cref{eq:dot-product-equality}.

($\Leftarrow$) Assume that $\operatorname{rank}(\Eprime) = |\Vdocids|$, meaning the span of $\Eprime$ is all of $\mathbb{R}^{|\Vdocids|}$, which, when viewed as the space of logits, produces all probability distributions over $\Vdocids$. In other words, for any $p \in \Delta^{|\Vdocids|-1}$ there exists a $\varphi^\prime \in \mathbb{R}^{n+1}$ such that $\sigma(\Eprime \varphi^\prime) = p$. But $\sigma(\Eprime \varphi^\prime) = \sigma(\Edocids \varphi^\prime_{1:n} + \varphi^\prime_{n+1}\mathbf{1}) = \sigma(\Edocids \varphi^\prime_{1:n})$, where the first equality follows from the definition of $\Eprime$ and the second from the shift-invariance of softmax. Since we assume our ARR has infinite capacity, given a context $c$ it can produce $\varphi(c) = \varphi^\prime_{1:n}$, and hence $\sigma(\Edocids \varphi(c)) = p$.
\end{proof}

\subsection{Proof of \Cref{prop:arbitrary-tok-level-ranking}}

\PropArbitraryTokenLevelRanking*

\begin{proof}
($\Rightarrow$) The proof is entirely analogous to the forward proof of \Cref{prop:arbitrary-tok-level-prob}, with the exception that the contradiction targets \textit{permutations} in which all token indices for which $h$ takes negative values \textit{precede} all token indices for which $h$ takes negative values.

($\Leftarrow$) This is an immediate consequence of \Cref{prop:arbitrary-tok-level-prob} — if the ARR can produce any probability distribution over $\Vdocids$, it can also produce any ranking over $\Vdocids$.
\end{proof}

\newpage
\section{Hyperparameter Settings}\label{sec:appendix-hparams}

\subsection{Generative}\label{sec:hparams-generative}

For the generative experiments, we take Mistral-7B-v0.3-it \citep{jiang2024identifying} as the base model for all fine-tuning experiments in this paper, with the following fine-tuning hyperparameters:
\begin{itemize}
    \item \textbf{Optimizer}: Adafactor \citep{shazeer2018adafactor} with $\beta_1=0.9$
    \item \textbf{Learning Rate}: $1 \times 10^{-5}$
    \item \textbf{Learning Rate Schedule}: linear warmup for 50 steps followed by a cosine decay.
    \item \textbf{Dataset-Specific Settings} ($10k$ steps):
    \begin{itemize}
        \item \textbf{WordNet}
        \begin{itemize}
            \item \textbf{Max Input length}: 1024
            \item \textbf{Global Batch Size}: 64, accumulated across replicas
            \item \underline{Top-1 docID} ($\lambda(r)=\mathbb{I}_{r=1}$): $7.65$ epochs over $84k$ examples
            \item \underline{Each docID} ($\lambda(r)=\frac{1}{r^\alpha}$ or $\lambda(r)=\frac{n_q - r + 1}{n_q}$): $0.91$ epochs over $700k$ examples
        \end{itemize}
        \item \textbf{ESCI}
        \begin{itemize}
            \item \textbf{Max Input Length}: 8192
            \item \textbf{Global Batch Size}: 32, accumulated across replicas
            \item \underline{Top-1 docID} ($\lambda(r)=\mathbb{I}_{r=1}$): $3.76$ epochs over $85k$ examples
        \end{itemize}
    \end{itemize}
    \item \textbf{Weight Decay}: no weight decay
    \item \textbf{Gradient Clipping}: gradient norm clipped to 1.0
\end{itemize}

\subsection{Dual Encoders}\label{sec:hparams-de}

For Dual Encoders we train for 50k steps with a batch size of 2k. We use SGD optimizer with a learning rate of 1.0 and momentum of 0.9. The embeddings are normalized by projecting each row of the embedding table to the unit $\ell_2$ sphere after each SGD update.

\subsection{Cross Encoders}\label{sec:hparams-ce}

For Cross Encoders we train for 20k steps with a batch size of 2k. 
We use ADAM optimizer with a learning rate of 0.001 and weight decay of 0.001. 
All weight matrices in MLP are initialized with an i.i.d. Gaussian of zero mean and a variance of $1/(2n)$.
All bias terms in MLP are initialized as zero vectors.

\newpage
\section{Dataset Statistics}\label{sec:dataset-statistics}

\begin{figure}[H]
    \centering
    \includegraphics[width=0.9\linewidth]{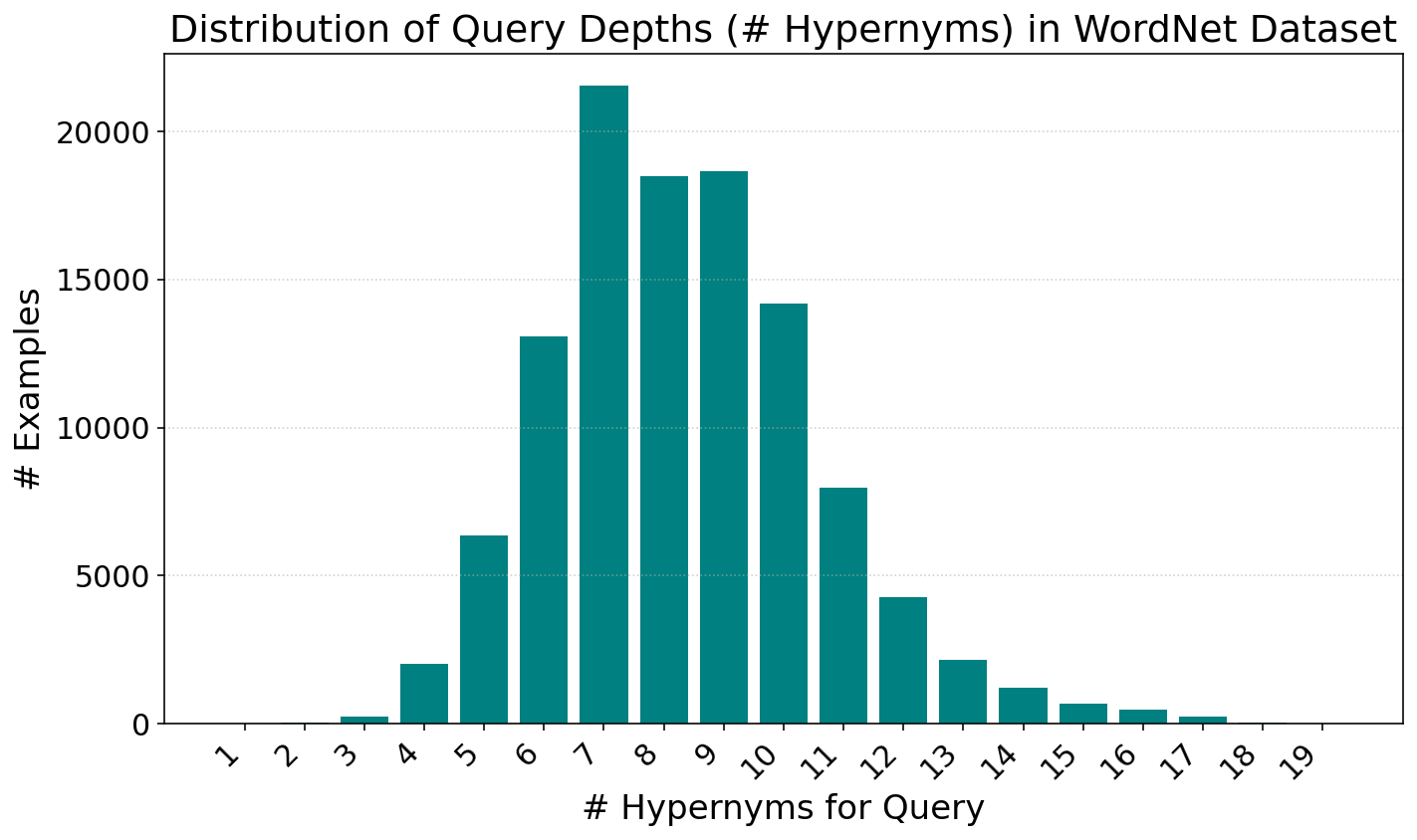}
    \label{fig:wordnet-distr}
\end{figure}    
\vspace{1cm}
\begin{figure}[H]
    \centering
    \includegraphics[width=0.9\linewidth]{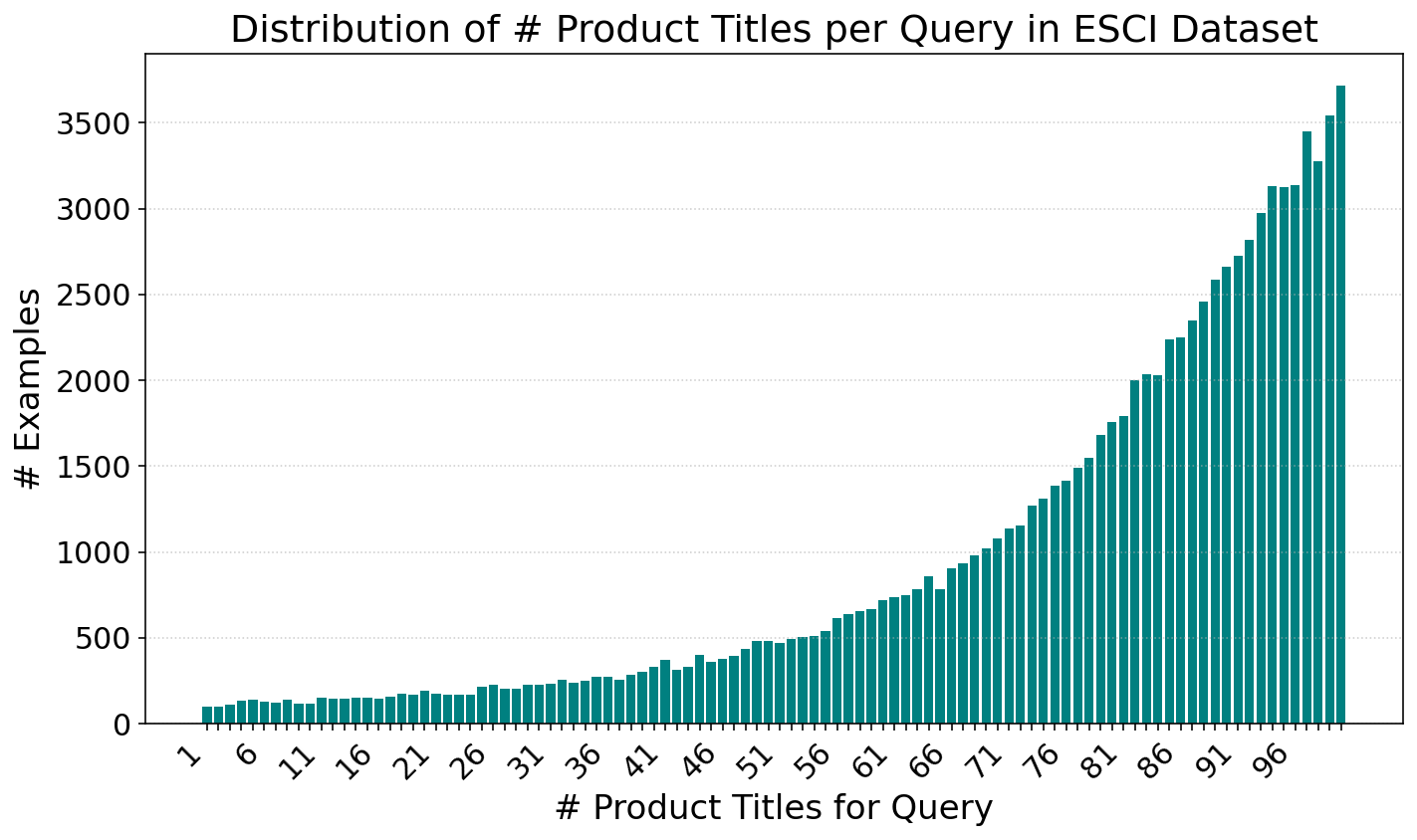}
    \label{fig:esci-distr}
\end{figure}

\newpage
\section{Zero-Shot Prompt Example}\label{sec:appendix-prompt}

\subsection{WordNet Prompt}\label{sec:wn-prompt}

\begin{tcolorbox}[promptbox]
\ttfamily
You are a helpful assistant that can answer questions about WordNet. You are given a noun synset from WordNet and a set of hypernyms for that synset. The hypernyms are separated by the delimiter `||'. Your task is to output the most specific hypernym for the synset out of the given set of hypernyms.

\vspace{0.5em} 

Synset:
deer.n.01

\vspace{0.5em} 

Hypernyms:
object.n.01 || whole.n.02 || chordate.n.01 || physical\_entity.n.01 || mammal.n.01 || ruminant.n.01 || vertebrate.n.01 || placental.n.01 || ungulate.n.01 || even-toed\_ungulate.n.01 || animal.n.01 || entity.n.01 || living\_thing.n.01 || organism.n.01

\vspace{0.5em} 

Most specific hypernym:
\end{tcolorbox}

\subsection{ESCI Prompt}\label{sec:esci-prompt}

\begin{tcolorbox}[promptbox]
\ttfamily
You are a helpful assistant that can answer questions about the ESCI shopping queries dataset. You are given a query and a list of product titles for that query. Each product title is preceded by its docID in square brackets, and the product titles are separated by the delimiter '||'. Your task is to generate the docID of the most relevant product title to the query out of the given set of product titles and their docIDs.

\vspace{0.5em} 

Query:
aukey usb c

\vspace{0.5em} 

Product titles:
[25,36,39] Syntech USB C to USB Adapter, 2 Pack USB C to USB3 Adapter,USB Type C to USB,Thunderbolt 3 to USB Female Adapter OTG Cable || [25,69,59] SSK USB C 10Gbps Hub, 4-in-1 SuperSpeed USB 10Gbps Type C Multiport Adapter with 2 USB C 2 USB A 3.1/3.2 Gen2 10Gbps Ports,USB C Dock for iMac/MacBook/Pro/Air/Surface Pro and More Type C Devices || [25,36,26] AUKEY Focus 63W USB C Charger 60W PD Charger Power Delivery 3.0 [ GaN Power ] Fast Charger Dual Port USB C Wall Charger for MacBook Pro Air 13" 15" iPhone 11, Pro,Max SE, Google Pixel 4 3 XL, Switch || [25,59,45] USB Type C Female to USB Male Adapter USB to Type C Laptops Adapter for iPhone 11 Pro Max Airpods iPad for Samsung Galaxy S20 Ultra Google Pixel 2Pack || [25,59,2] USB C Cable,Aupek 10FT USB Type C Cable Nylon Braided Fast Charger Cord for Nintendo Switch, Google Pixel,Samsung Galaxy Note 8 S9 S8 S8 Plus S9(Purple) || [25,39,59] AUKEY USB Car Charger, 18w Quick Charge 3.0, Flush Fit Cell Phone Adapter for iPhone 12 Pro Max/11 Pro Max/XS/XR, Samsung Galaxy Note 9 / S9 / Note 10 / S10, and More || [26,2,77] AUKEY Mechanical Keyboard Blue Switch, 104-Key RGB Backlit Gaming Keyboard with Customizable Lighting Effects, Aluminium USB Wired Keyboard for Gaming and Typing || [25,2,46] AUKEY Power Strip with 4 AC Outlets and 4 USB Charging Ports, 5-Foot Extension Cord for Smartphone, Laptop, Tablet, Home, Office and More (White) || [25,36,97] AUKEY Focus iPhone Fast Charger 30W 2-Port USB C Charger for iPhone 12/12 Mini/12 Pro Max, PD 3.0 Fast Charger, USB C Wall Charger for iPhone 11 Pro Max/8 Plus, Pixel 5, MacBook Air, iPad Pro, Switch || [36,25,31] USB C Cable, 2Pack 3ft 5ft || [25,26,46] AUKEY USB C Cargador con GAN, Cargador de Pared USB con 60W Power Delivery 3.0, Compatible con MacBook Pro 13", iPhone 11 Pro, DELL XPS 13, HP Spectre, Lenevo Thinkpad, Nintendo Switch || [25,36,45] AUKEY USB C Adapter, [3 Pack] USB C to USB 3.0 Adapter Compatible with MacBook Pro 2017/2016 , Google Chromebook Pixelbook , Samsung Galaxy S9 S8 S8+ Note8, Google Pixel 2/2XL - Black || [25,36,39] USB C to USB Adapter (3 Pack), Warmstor Type C Female to USB 3.0 A Male Converter Support Data Sync \& Charging Compatible with iPhone, iPad, Samsung, Pixel, Laptops, PC, Power Banks, Chargers

\vspace{0.5em} 

Most relevant product title:
\end{tcolorbox}

\newpage
\section{WordNet: \Cref{tab:wordnet-results} Visualizations}\label{wordnet:comparative-visualizations}

\begin{figure}[H]
    \centering

    \begin{subfigure}[b]{0.48\textwidth}
        \includegraphics[width=\textwidth]{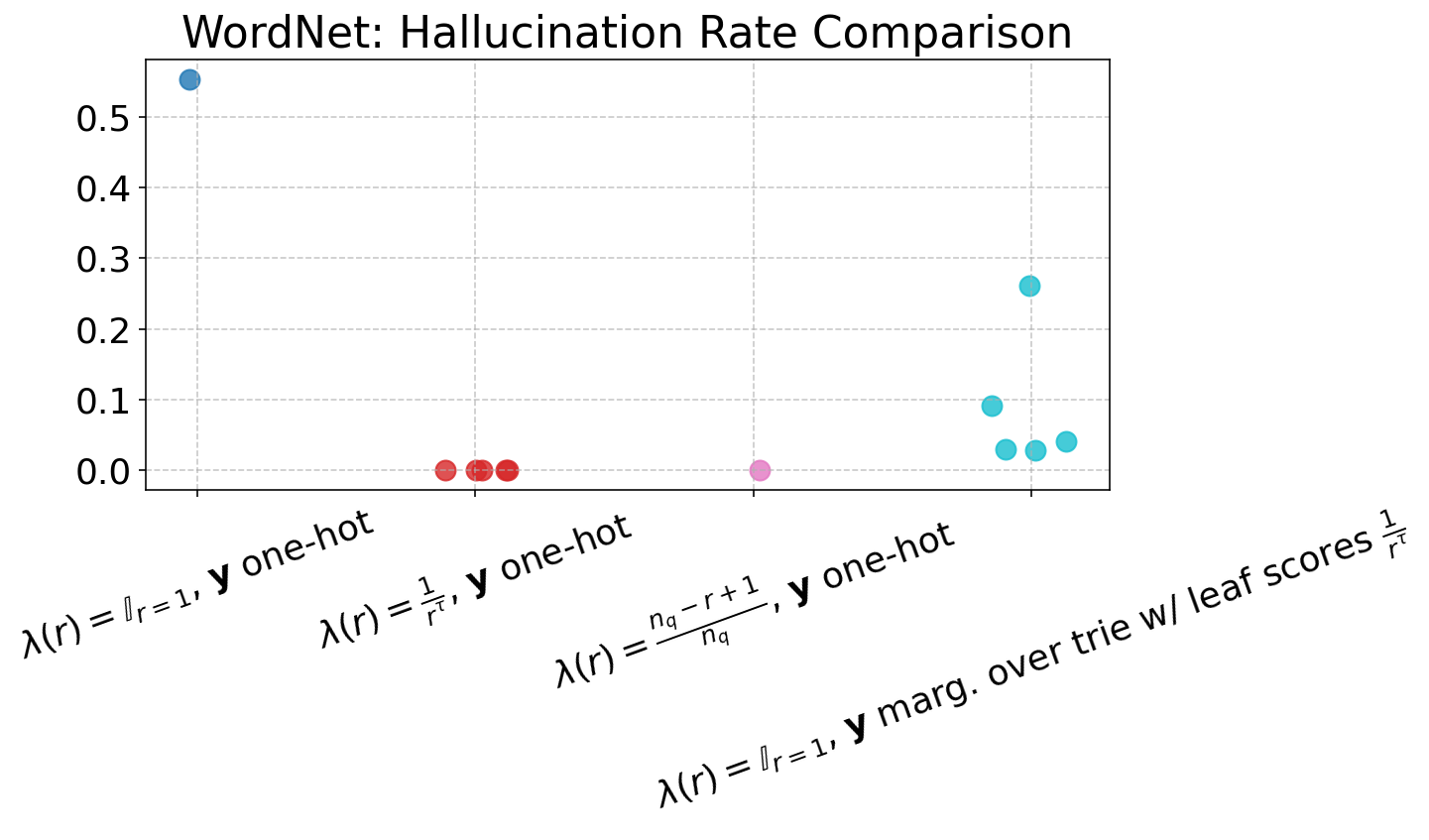}
    \end{subfigure}
    \hfill
    \begin{subfigure}[b]{0.48\textwidth}
        \includegraphics[width=\textwidth]{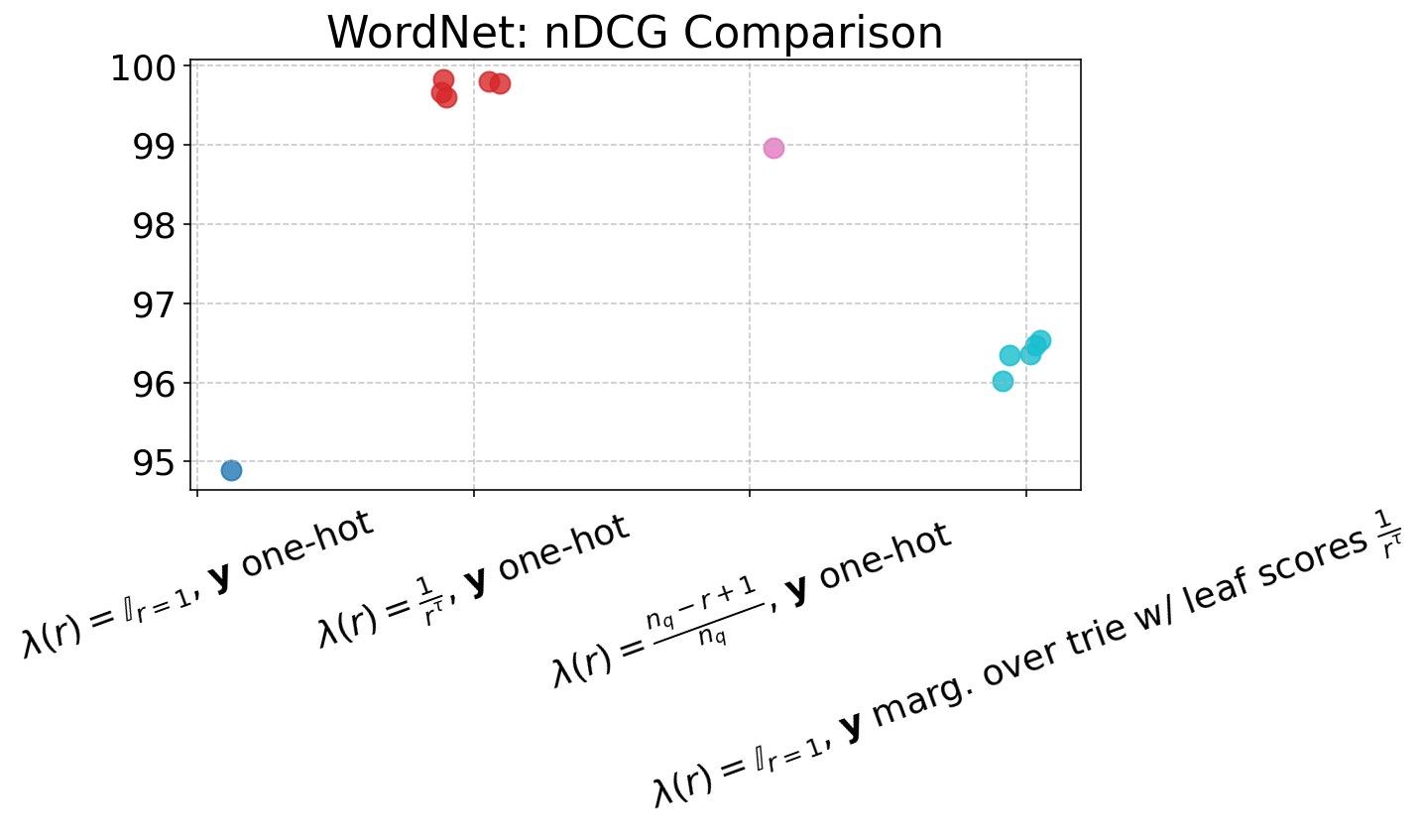}
    \end{subfigure}

    \vspace{1em}

    \begin{subfigure}[b]{0.48\textwidth}
        \includegraphics[width=\textwidth]{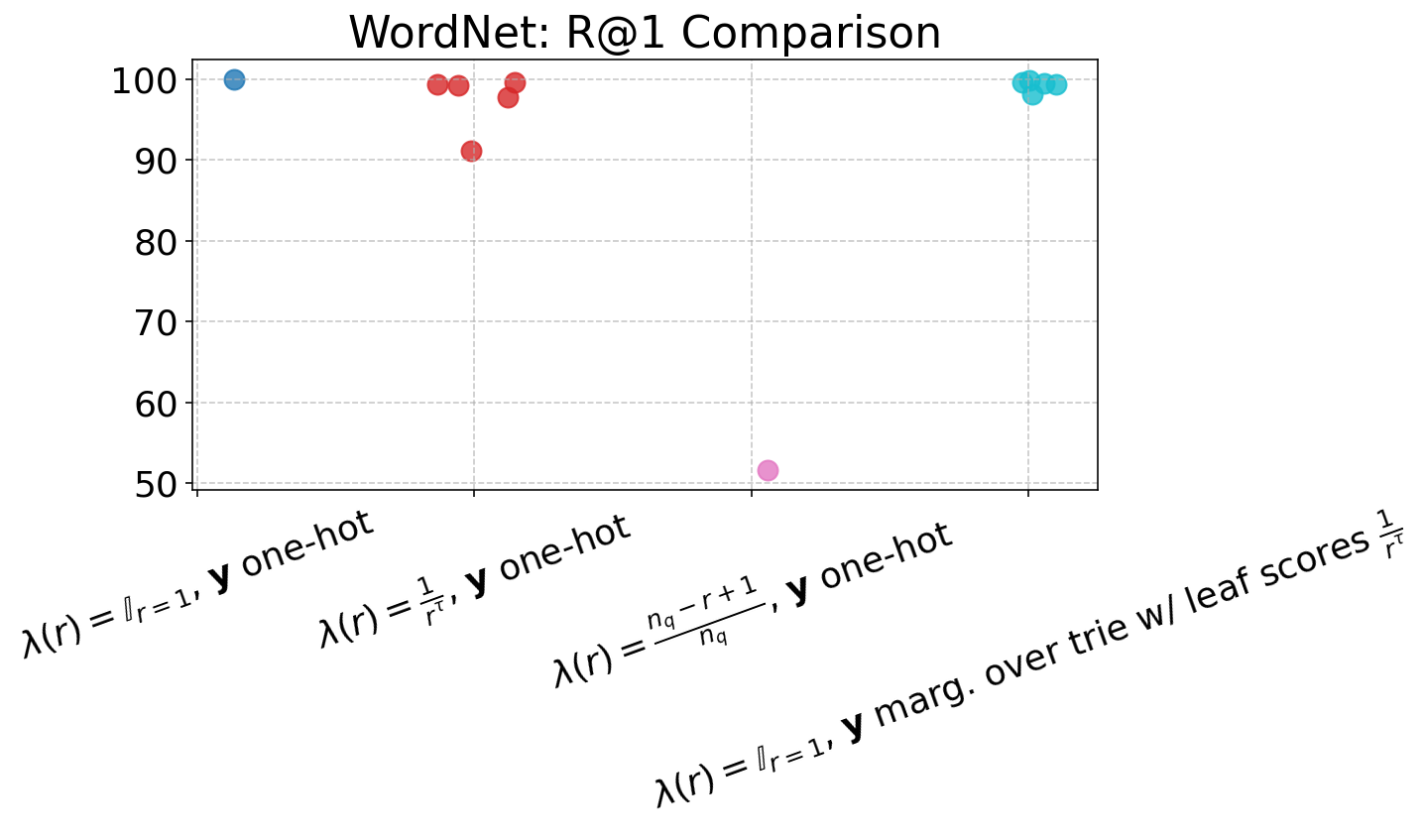}
    \end{subfigure}
    \hfill
    \begin{subfigure}[b]{0.48\textwidth}
        \includegraphics[width=\textwidth]{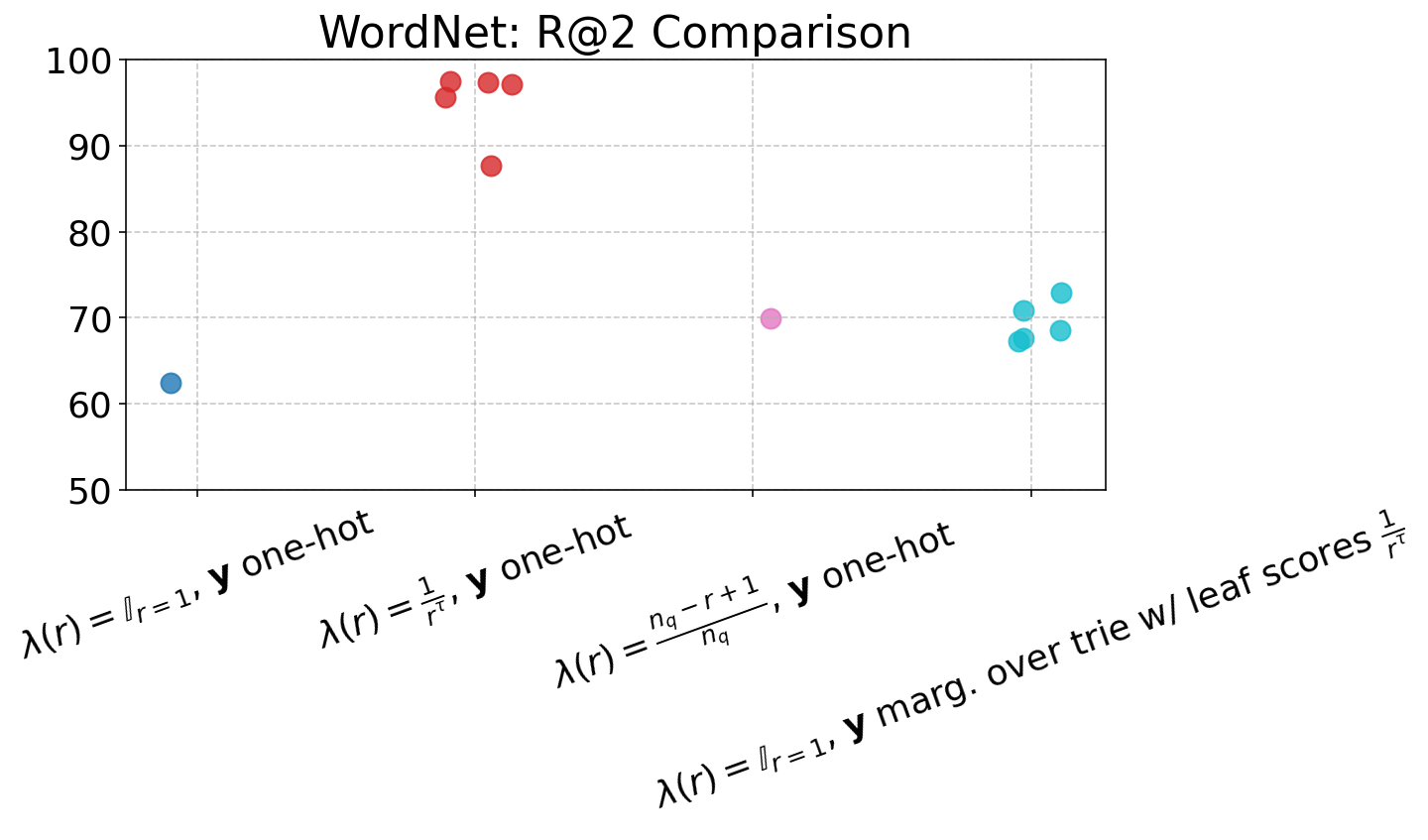}
    \end{subfigure}

    \vspace{1em}

    \begin{subfigure}[b]{0.48\textwidth}
        \includegraphics[width=\textwidth]{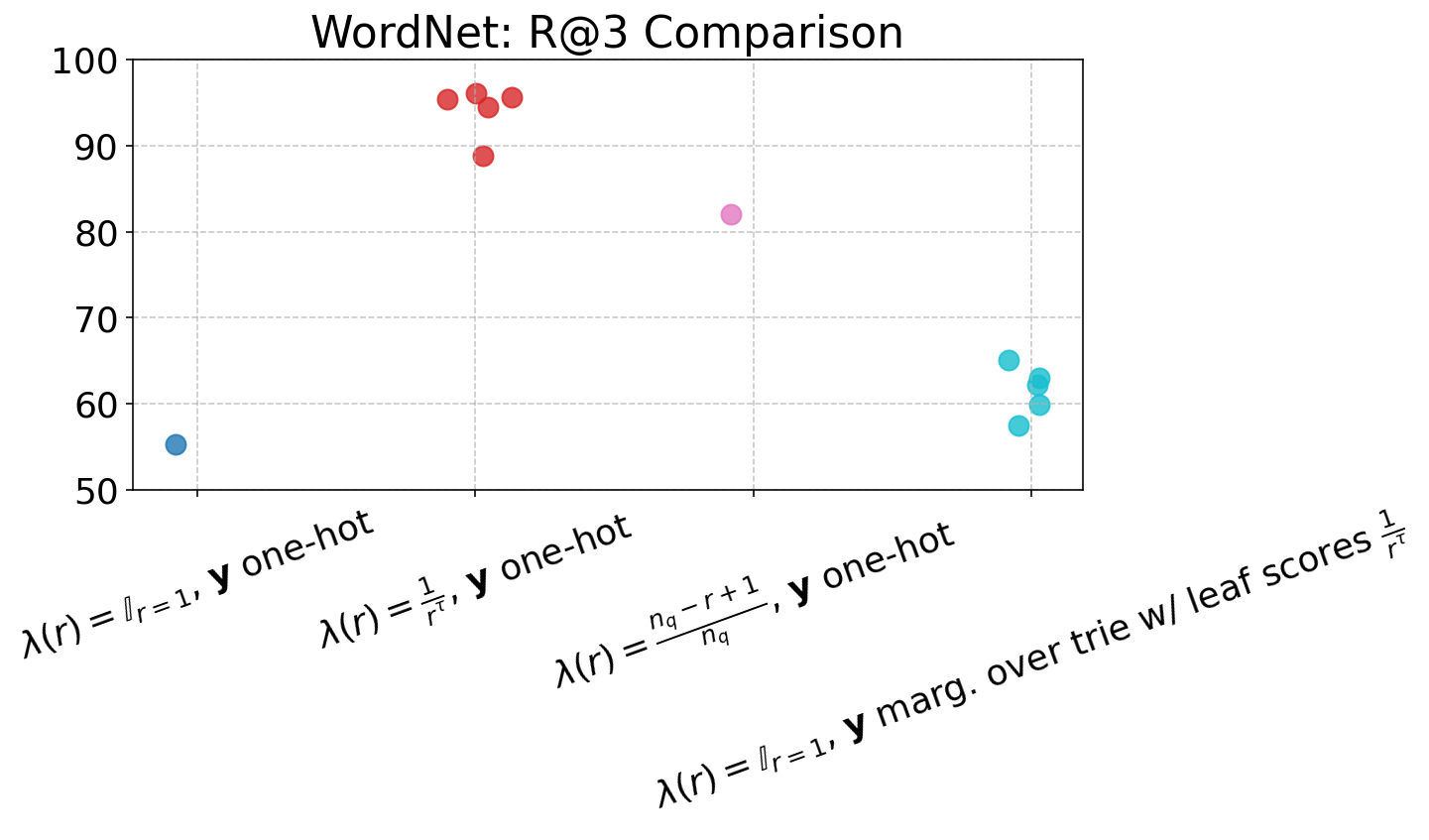}
    \end{subfigure}
    \hfill
    \begin{subfigure}[b]{0.48\textwidth}
        \includegraphics[width=\textwidth]{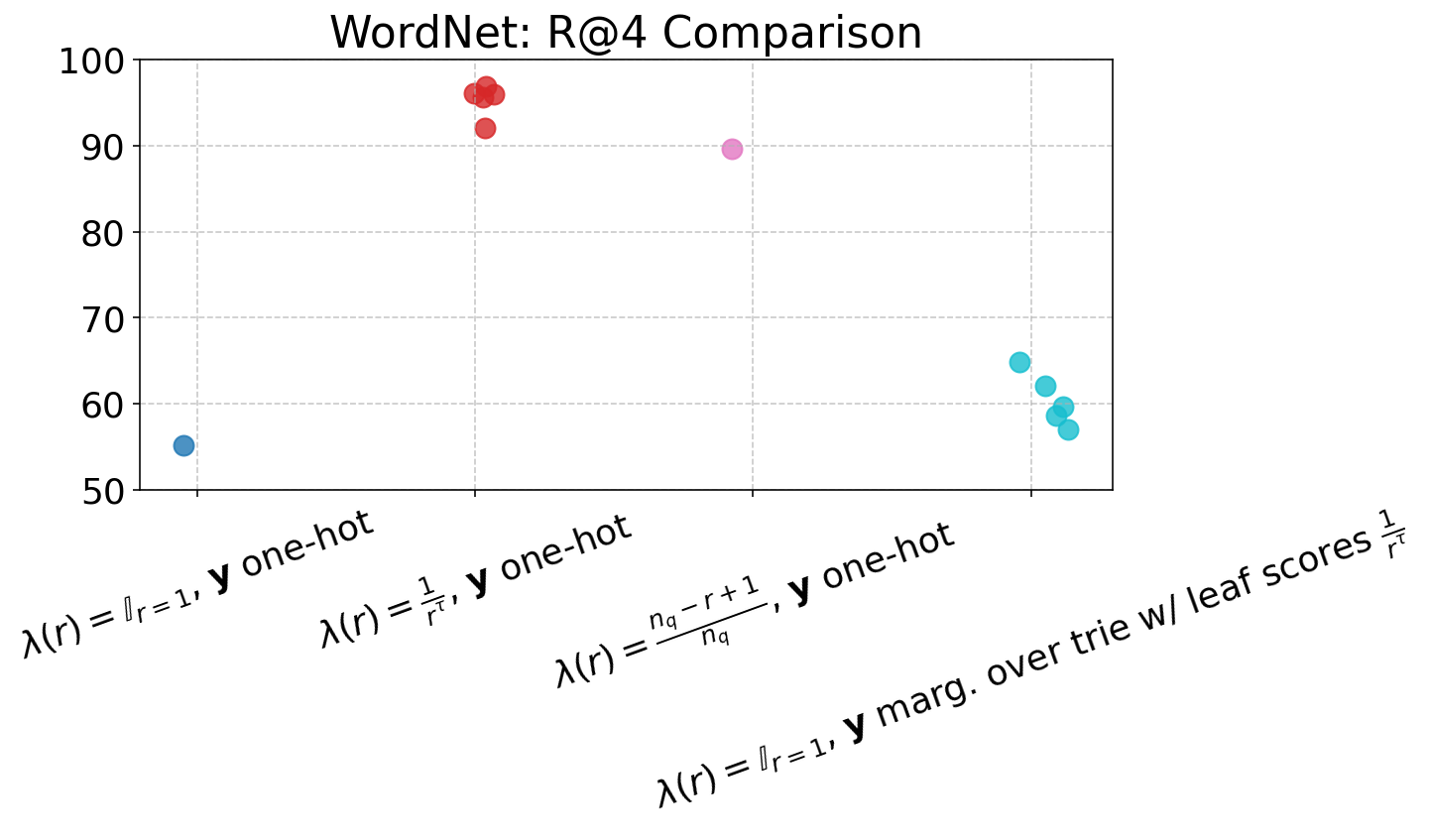}
    \end{subfigure}

    \vspace{1em}

    \begin{subfigure}[b]{0.48\textwidth}
        \includegraphics[width=\textwidth]{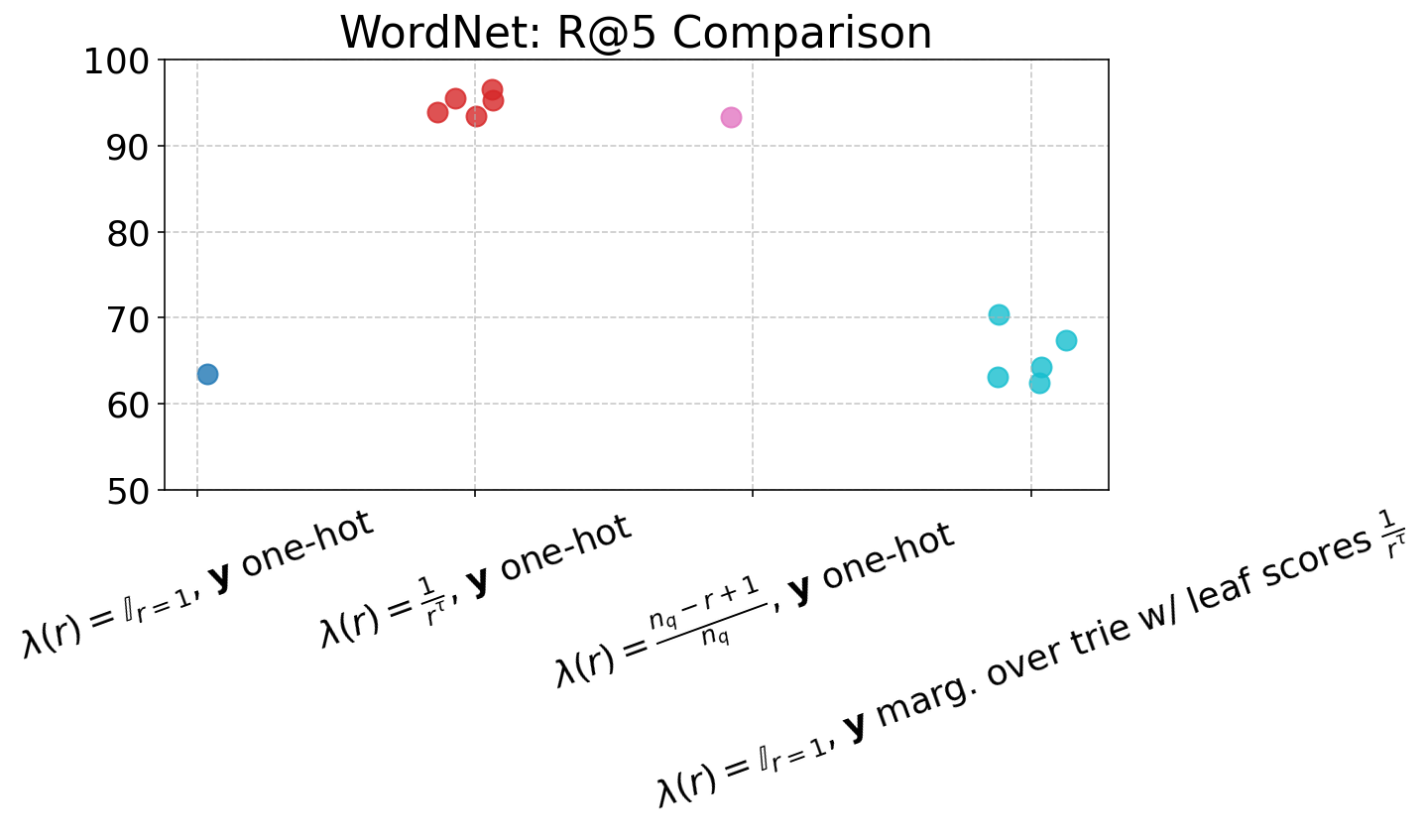}
    \end{subfigure}
\end{figure}

\newpage
\section{WordNet: Combined Item-and-Token Loss Results}

We provide supplemental results for fractional item-level reweighting \textit{combined with} the marginalized-over-trie target distributions $\mathbf{y}(r,t)$, i.e., leveraging the full expressivity of \Cref{eq:general-loss-formulation}. We take the temperature $\alpha=2$ which yielded best performance for R@4 and R@5 in \Cref{tab:wordnet-results}, and sweep over a range of $\beta$ temperatures for assigning scores to docID nodes in the prefix tree to be summed over.

In order to construct $\mathbf{y}(r,t)$ for docID of rank $r$, we only use docIDs $d_r, d_{r+1}, \ldots, d_{n_q}$ in the trie creation process. For example, if we are considering the contribution of the $3^\text{rd}$-ranked docID, which is already down-weighted by $\lambda(3)=\frac{1}{3^\alpha}$, we want $\mathbf{y}(r,t)$ to smooth over the top-$3^\text{rd}$, top-$4^\text{th}$, top-$5^\text{th}$ docID paths in the prefix tree without the target distributions being skewed by the top-$1^\text{st}$ and top-$2^\text{nd}$ docIDs.

\begin{table}[htbp]
\small
\centering
\resizebox{\textwidth}{!}{%
\begin{tabular}{c c ccccccc}
\toprule
& & CVR ($\downarrow$) & nDCG ($\uparrow$) & R@1 ($\uparrow$) & R@2 ($\uparrow$) & R@3 ($\uparrow$) & R@4 ($\uparrow$) & R@5 ($\uparrow$) \\
\midrule
\multirow{9}{*}{\parbox{3.2cm}{\centering $\lambda(r) = \frac{1}{r^2}$ \\ $\mathbf{y}$ marg. over trie w/ leaf scores $\frac{1}{r^\tau}$}}
& $\beta=1$ & 0.0\% & 99.12 & 73.38 & 83.37 & 86.83 & 90.08 & 91.51 \\
& $\beta=2$ & 0.02\% & 99.74 & 92.36 & 94.45 & 94.73 & 95.39 & \textbf{96.11} \\
& $\beta=3$ & 0.0\% & 99.71 & 94.28 & 94.1 & 94.49 & 95.46 & 96.06 \\
& $\beta=4$ & 0.0\% & 99.7 & 94.9 & 94.3 & 94.28 & 95.53 & 95.67 \\
& $\beta=5$ & 0.0\% & 99.74 & 95.92 & 94.54 & 95.17 & 95.38 & 95.7 \\
& $\beta=6$ & 0.0\% & 99.77 & 97.68 & \textbf{96.29} & \textbf{95.67} & 95.7 & 95.69 \\
& $\beta=7$ & 0.0\% & \textbf{99.83} & 97.82 & 95.74 & 94.83 & \textbf{95.79} & 95.93 \\
& $\beta=9$ & 0.0\% & 99.76 & \textbf{98.32} & 95.72 & 94.86 & 94.85 & 95.19 \\
& $\beta=25$ & 0.0\% & 99.71 & 97.38 & 94.42 & 94.64 & 94.82 & 95.05 \\
\bottomrule
\end{tabular}
}
\caption{WordNet results for combined item-and-token loss (evaluated over 5000 (query, targets)-examples)}
\label{tab:wordnet-combined-loss-from-checkpoint}
\end{table}

Note that as $\beta \rightarrow \infty$, the target distributions $\mathbf{y}(r,t)$ looks more and more like one-hots and the results in \Cref{tab:wordnet-combined-loss-from-checkpoint} should converge to those obtained with just item-level reweighting $\lambda(r)=\frac{1}{r^\alpha}, \alpha=2$ in \Cref{tab:wordnet-results}. While the nDCG score in \Cref{tab:wordnet-combined-loss-from-checkpoint} improves as $\beta$ increases from 1 to 7, unfortunately this trend seems to stop for larger $\beta$ and we do not see it surpass the $\lambda(r)=\frac{1}{r^2}$ results with $\mathbf{y}$ one-hots from \Cref{tab:wordnet-results} at any point. We leave optimizing for the ``sweet spot'' combination of item-level and token-level hyperparameters to future work.


\end{document}